\documentclass[a4paper,10pt]{article}
\usepackage[margin=25mm]{geometry}
\usepackage{amsmath,amssymb,amsthm,mathtools}
\usepackage{mathrsfs}
\usepackage[authoryear,round]{natbib}
\usepackage{booktabs,longtable,array}
\usepackage{graphicx}
\usepackage{placeins}
\usepackage[colorlinks=true,linkcolor=blue,citecolor=blue,urlcolor=blue]{hyperref}
\usepackage{pgfplots}
\pgfplotsset{compat=1.17}
\usepackage{enumitem}
\setlist{topsep=4pt,itemsep=3pt}
\numberwithin{equation}{section}
\newtheorem{theorem}{Theorem}[section]
\newtheorem{proposition}[theorem]{Proposition}
\newtheorem{corollary}[theorem]{Corollary}
\theoremstyle{definition}

\newtheorem{assumption}[theorem]{Assumption}
\newtheorem{example}[theorem]{Example}
\theoremstyle{definition}
\newtheorem{remark}[theorem]{Remark}

\def\ve{{\varepsilon}}

\def\l{\left}
\def\r{\right}

\def\b{\beta}

\def\th{\theta}

\def\s{\sigma}

\def\k{\kappa}
\def\la{\lambda}

\def\vp{\varphi}

\def\p{\partial}

\def\A{\mathcal{A}}

\def\P{{\mathbb{P}}}

\def\E{\mathbb{E}}

\def\I{\mathbf{1}}
\def\df{\mathrm{d}}

\def\ol#1{\overline{#1}}
\def\ul#1{\underline{#1}}

\title{Generalized survival energy models with exogenous covariates for mortality prediction}
\author{Yasutaka Shimizu\thanks{Corresponding author. Email: \href{mailto:shimizu@waseda.jp}{shimizu@waseda.jp}.}\\Department of Applied Mathematics, Waseda University\\3-4-1, Okubo, Shinjuku-ku, Tokyo 169-8555, Japan}
\date{}
\begin{document}
\maketitle
\begin{abstract}
We develop a generalized survival energy model that relates mortality to social and economic conditions through a latent energy process. Death occurs when this process first reaches or crosses zero. Environmental covariates enter the coefficients governing depletion, recovery, fluctuations and accumulated damage, yielding lifetime distributions conditional on the environmental path. We derive these distributions for diffusion and inverse Gaussian specifications and describe their numerical evaluation. An application to Japanese male and female mortality data assesses predictive performance against Cairns--Blake--Dowd benchmarks. Conditional on realized environmental paths, the selected inverse Gaussian models achieve small prediction errors comparable to those of the benchmarks. The framework also provides a model-based interpretation of how environmental changes affect mortality through the rate of damage accumulation or the magnitude of energy fluctuations. These interpretations describe conditional statistical relationships rather than causal effects. The resulting models offer a way to study mortality under alternative environmental scenarios while retaining a connection to classical actuarial survival models.
\end{abstract}

\section{Introduction}\label{sec:randominitial}
Mortality can be modelled either directly through age-specific death rates or indirectly through the evolution of an individual's latent vitality. \citet{SMI2021} proposed the survival energy model (SEM), in which death occurs when a latent stock of survival energy is exhausted. Survival energy is a modelling construct that represents depletion, recovery and individual randomness, rather than a directly observed biological quantity. Its stochastic evolution determines the lifetime distribution through the first passage to zero or below.

The interpretation of mortality in terms of declining vitality has a long history in actuarial science. \citet{Gompertz1825} derived an exponentially increasing force of mortality from the idea that resistance to death declines at a constant proportional rate. \citet{Minami1987} discussed genetically programmed ageing and cumulative damage theories, and developed a mortality model based on the latter. \citet{ZZW2026} review vitality-based mortality models. The survival energy hypothesis can be viewed as part of this broader tradition of explaining lifetimes through the depletion of vitality or the accumulation of damage.

\subsection{A generalized survival energy model}
To extend the SEM of \citet{SMI2021}, we distinguish initial energy, systematic depletion or recovery, continuous fluctuations and accumulated damage. \citet{ZZW2026,ZZW2024} also consider a general framework combining initial vitality, depletion, diffusion and jumps. Building on these ideas, we introduce environmental dependence through the generalized survival energy model (GSEM)
\begin{equation}
 dX_t^c=\left\{-\frac{\mu_0(t)}{\rho}+a(t,Z_t^c)\right\}dt
       +\sigma(t,Z_t^c)\,dW_t^c-dD_t^c,
 \qquad X_0^c=\xi^c\sim\operatorname{Exp}(\rho)
 \label{GSEM}
\end{equation}
A cohort consists of individuals of the same sex and birth year. The sexes are modelled separately; $c$ indexes birth year, $t$ denotes age, and $Z_t^c\in\mathbb R^d$ collects environmental covariates. Here $\mu_0(t)\ge0$ is a baseline force of mortality, $a$ adjusts depletion and may allow recovery, $\sigma$ controls fluctuations driven by a standard Brownian motion $W^c$, and $D^c$ is a nondecreasing damage process. The parameter $\rho>0$ is the rate of the exponential initial energy distribution.

Define the lifetime by $\tau_c=\inf\{t\ge0:X_t^c\le0\}$, the cumulative probability of death by $F_c(t)=\Pr(\tau_c\le t)$, and the survival function by $S_c(t)=1-F_c(t)$. We refer to $F_c$ as the mortality distribution function (CDF); it corresponds to the mortality function of \citet{SMI2021}. The framework encompasses inhomogeneous diffusion (ID) models and inverse Gaussian (IG) damage models~\citep{SSKMI2023,MS2024}, as well as specifications with fixed initial energy. Section~3 develops these models and their lifetime distributions.

\subsection{Random initial energy and classical survival models}
Random initial energy represents individual heterogeneity and establishes a connection with classical mortality models. Setting $a=\sigma=0$ and $D^c\equiv0$ in \eqref{GSEM} gives
\[
 X_t^c=\xi^c-\frac1\rho\int_0^t\mu_0(s)\,ds
\]
If $\mu_0$ is nonnegative and locally integrable, energy cannot recover. The exponential tail therefore yields
\begin{equation}
 S_c(t)=\Pr\!\left(\xi^c>\frac1\rho\int_0^t\mu_0(s)\,ds\right)
       =\exp\!\left\{-\int_0^t\mu_0(s)\,ds\right\}
\end{equation}
This is the usual survival function associated with the force of mortality $\mu_0$. For example, $\mu_0(t)=m+ge^{kt}$ gives the Gompertz--Makeham law. Thus GSEM retains the classical survival representation while allowing fluctuations, accumulated damage and environmental effects.

\subsection{Environmental covariates and scope of the study}
Age $t$ in cohort $c$ corresponds to calendar year $c+t$. Annual GDP and physician density, for example, can be collected in a series $E_u$ and incorporated as $Z_t^c=E_{c+t}$. We distinguish the shared environment from individual randomness by assuming that the environmental process is independent of the underlying individual random inputs. Initial energy is also independent of the environment and the driving noises. Individuals within a cohort share the same model law. Environmental paths across cohorts are not assumed independent, since they arise from common calendar-time series.

Conditional on an environmental path $z$, coefficients such as $a_z(t)=a(t,z(t))$ and $\sigma_z(t)=\sigma(t,z(t))$ determine the lifetime distribution $F(t\mid z)$. The damage process may also depend on the environment. This specification links changes in mortality to changes in depletion, recovery, fluctuations or damage accumulation. Exogenous covariates have been used in mortality prediction before~\citep{NiuMelenberg2014,Dimai2025}; our focus is on representing their effects through the dynamics of survival energy.

We use the term \emph{structural model} for GSEM, in contrast to \emph{reduced-form models} that specify mortality directly. We fit fixed-initial-energy IG and ID models to Japanese cohort mortality data for each sex and compare them with the CBD-family models M6 and M7~\citep{Cairns2009}. In an evaluation conditional on realized environmental paths, the IG models achieve predictive accuracy comparable to these benchmarks. We also examine alternative environmental paths to interpret their model-implied mortality effects. These calculations are not estimates of causal effects; they assess the framework in terms of both predictive performance and interpretability.

Section~2 describes potential environmental data sources. Section~3 derives lifetime distributions, and Section~4 specifies the empirical models and estimation methods. Section~5 reports predictive performance and environmental scenario results. Section~6 concludes.

\section{Publicly available environmental covariates}

Access to medical care, economic security and exposure to pollution or extreme temperatures are potential indicators of the environment relevant to survival energy. We summarize public data sources for Japan and the main issues in their use. World Bank coverage was checked on 29 September 2026. Section~\ref{sec:paperig} specifies the two covariates used in the empirical analysis and the treatment of missing observations.

\begin{longtable}{p{0.19\textwidth}p{0.34\textwidth}p{0.39\textwidth}}
\toprule
Domain & Interpretation & Sources and limitations\\\midrule
\endhead
Medical resources and access & Physician and nurse density, access to care and emergency response times. & Physician density is available from the World Bank; OECD and e-Stat provide additional workforce, coverage and prefectural data~\citep{WBdata,OECDhealth,ESTATdoctors}. Resource quantity is not equivalent to quality; more beds need not imply lower mortality.\\
Economic conditions & Real income per capita, unemployment, poverty, inequality and income volatility. & Real GDP per capita at purchasing power parity (PPP) and unemployment are available from the World Bank~\citep{WBdata}. GDP does not measure household income or its distribution. Consistent long-run poverty and inequality series were not verified for this study.\\
Air quality and climate & PM2.5 exposure and extreme heat or cold. & PM2.5 is available from the World Bank; station-level weather data are available from the Japan Meteorological Agency~\citep{WBdata,JMAdata}. National exposure differs from individual exposure. Station relocation and missing values require attention.\\
Sanitation and housing & Safe sanitation, insulation, heating, cooling and crowding. & Safely managed sanitation coverage is available from the World Bank~\citep{WBdata}, but its high level and limited variation in Japan may impede estimation. Comparable long-run housing series were not verified.\\
Health expenditure & Per-capita expenditure and public financing as proxies for available services. & PPP-adjusted expenditure is available from the World Bank~\citep{WBdata}. Expenditure may increase in response to morbidity. The retrieved series uses current international dollars and requires price adjustment for long-run comparisons.\\
Lifestyle and public health & Smoking, nutrition, exercise, vaccination and preventive care. & OECD and the National Health and Nutrition Survey provide relevant data~\citep{OECDhealth,NHNS}. Age, sex and year coverage must be aligned. Behaviour and health status may make these variables endogenous.\\
Social support and institutions & Isolation, family support, long-term care access and employment or social protection. & These remain conceptual candidates: comparable long-run data for the study cohorts were not verified. Overlap with income and access-to-care measures should be avoided.\\
\bottomrule
\end{longtable}

\begin{remark}[Covariate selection]\label{rem:covariates}
The components of $Z$ should be selected on substantive grounds and with adequate temporal coverage. Physician density may proxy access to care, income may reflect economic security, and pollution may indicate physical exposure. Including several closely related indicators can make their separate effects difficult to identify.
\end{remark}

\begin{example}[Annual series available for Japan]\label{ex:dataavailability}
The following Japanese series were retrieved from the World Bank~\citep{WBdata}. Coverage gives the earliest and latest retrieved years, and the final column counts available annual observations. Coverage need not be complete: physician density, for example, has only 39 observations over 1960--2022. Retrieval dates and data versions should be retained because coverage and values are revised.
\begin{center}
\small
\begin{tabular}{p{0.32\textwidth}p{0.32\textwidth}rr}
\toprule
Indicator & API code & Coverage & Count\\\midrule
Physicians per 1,000 people & \texttt{SH.MED.PHYS.ZS} & 1960--2022 & 39\\
Health expenditure per capita (PPP, current prices) & \texttt{SH.XPD.CHEX.PP.CD} & 2000--2023 & 24\\
GDP per capita (PPP, constant 2021 prices) & \texttt{NY.GDP.PCAP.PP.KD} & 1990--2025 & 36\\
Mean annual PM2.5 exposure & \texttt{EN.ATM.PM25.MC.M3} & 1990--2023 & 34\\
Unemployment (ILO modelled estimate) & \texttt{SL.UEM.TOTL.ZS} & 1991--2025 & 35\\
Safely managed sanitation coverage & \texttt{SH.STA.SMSS.ZS} & 2000--2024 & 25\\
\bottomrule
\end{tabular}
\end{center}
The API codes identify the series. Definitions should be checked carefully, as observations may include estimates and adjustments for international comparability.
\end{example}

\begin{example}[Constructing a covariate vector]\label{ex:initialZ}
One illustrative specification combines medical resources, economic conditions and air quality:
\[
 Z_t^c=\bigl(\operatorname{std}(\log D_{c+t}),\,
 \operatorname{std}(\log Y_{c+t}),\,
 -\operatorname{std}(P_{c+t})\bigr)^\top
\]
Here $D_y$ is physician density, $Y_y$ is real PPP-adjusted GDP per capita, and $P_y$ is PM2.5 exposure in year $y$. The operator $\operatorname{std}$ subtracts the training-period mean and divides by the training-period standard deviation. The same transformation is applied during prediction. The minus sign on PM2.5 makes larger values correspond to lower pollution. This coding does not establish that higher income or physician density necessarily improves health. Alternative indicators, such as unemployment or heat exposure, could be used in sensitivity analyses.

This three-variable example illustrates construction of $Z$; the empirical analysis below uses only GDP and physician density. Missing annual values require an explicit interpolation rule. Moreover, similar trends in covariates and mortality may reflect common period effects. Separating environmental associations from age and cohort effects is therefore an identification issue.
\end{example}

\begin{remark}[Annual inputs and exogeneity]\label{rem:covariatecautions}
Annual covariates are aligned with cohort age through $Z_t^c=E_{c+t}$. Delayed or cumulative effects could instead be represented by $E_{c+t-r}$, $r\ge0$, or $\int_0^t w(t-s)E_{c+s}\,\df s$, with an explicit weighting function $w$, units and interpolation rule. Regional analyses also require geographic alignment. The Japanese Mortality Database provides national and prefectural data~\citep{JMDdata}, but migration can disconnect past exposure from current residence.

Environmental series often cover shorter periods than mortality data. Reconstructing exposure from birth for older cohorts may require historical sources, a restricted observation period or a separate model for energy at entry. Our empirical analysis also imputes early environmental values, making the results dependent on those assumptions.

Environmental indicators may themselves respond to health. Health expenditure, for example, can rise because morbidity increases. Life expectancy or mortality measures derived from the response data are excluded to avoid feeding the outcome back into the predictors. Public availability or lagging alone does not establish the independence assumption in Section~1. The model component affected by each covariate must be specified, and statistical associations must be distinguished from causal effects.
\end{remark}

\section{First-passage distributions under GSEM}\label{sec:computation}
We derive lifetime distributions for the GSEM introduced in Section~1 with $\xi\sim\operatorname{Exp}(\rho)$. Coefficients may depend on both age and the environment; the empirical specifications are given in Section~\ref{sec:paperig}. The IG model describes nondecreasing damage, whereas the ID model allows energy to fluctuate continuously. Appendix~\ref{app:generaldistribution} treats the general jump-diffusion framework containing both cases.

Fix an environmental path $z$ and suppress the cohort index $c$. Let $\th$ denote the model parameters and write $\P_\th^z$ and $\E_\th^z$ for conditional probability and expectation. Parameter dependence is suppressed where convenient, but environmental dependence is retained through a subscript, superscript or conditioning argument. For the conditional process $X^z$ and lifetime $\tau_z$, define
\[
 F_\th(t\mid z)=\P_\th^z(\tau_z\le t),\qquad S_\th(t\mid z)=1-F_\th(t\mid z)
\]
If the density $f_\th(t\mid z)=\partial_tF_\th(t\mid z)$ exists, the force of mortality is $\mu_\th(t\mid z)=f_\th(t\mid z)/S_\th(t\mid z)$ whenever $S_\th(t\mid z)>0$. Here $\partial_t$ denotes differentiation with respect to $t$. Initial energy is independent of the environment and the driving random inputs.
The CDF $F(t\mid z)$ integrates over initial heterogeneity, whereas $F(t\mid h,z)$ conditions additionally on initial energy $h>0$. They satisfy
\[
 F(t\mid z)=\int_0^\infty F(t\mid h,z)\rho e^{-\rho h}\,\df h
\]
We also give formulas for fixed initial energy, which is used in the empirical analysis. A fixed initial value corresponds to a point-mass initial distribution, not to sending the exponential rate $\rho$ to infinity.

\subsection{IG models: baseline depletion and accumulated damage}\label{sec:igdistribution}
Consider the GSEM without diffusion,
\begin{equation}
 X_t^z=\xi-C_z(t)-D_t^z,\qquad
 C_z(t)=\int_0^t\{\mu_0(u)/\rho-a_z(u)\}\,\df u,
 \qquad D_t^z=Y^{(\k)}_{\Lambda_z(t)}
 \label{eq:generalig_model}
\end{equation}
which extends the IG-based SEM~\citep{SSKMI2023,MS2024} by adding initial heterogeneity and baseline depletion. The cumulative continuous depletion $C_z$ is nondecreasing under the restriction $\mu_0(u)/\rho-a_z(u)\ge0$. Both $C_z$ and $D^z$ are then nondecreasing. In particular, setting $a_z=0$ allows any nonnegative baseline force of mortality $\mu_0$.

The process $Y^{(\k)}$ is a nondecreasing IG process starting at zero, with stationary independent increments satisfying, for $\k>0$,
\[
 Y^{(\k)}_u-Y^{(\k)}_v\sim\operatorname{IG}(u-v,\k(u-v)^2),\qquad 0\le v<u
\]
Here $\operatorname{IG}(m,k)$ denotes the inverse Gaussian distribution with mean $m$ and variance $m^3/k$. Thus $Y^{(\k)}_u$ has mean $u$ and variance $u/\k$. Assume that $\Lambda_z$ is absolutely continuous and nondecreasing, with $\Lambda_z(0)=0$. The time change $\Lambda_z(t)$ incorporates age and environmental effects on damage accumulation. Consequently, $D_t^z$ has mean $\Lambda_z(t)$ and variance $\Lambda_z(t)/\k$.

Monotonicity implies that survival to age $t$ is equivalent to initial energy exceeding cumulative loss at that age. The exponential tail gives
\[
 S(t\mid z)=\P\{\xi>C_z(t)+D_t^z\mid z\}
 =e^{-\rho C_z(t)}\E[e^{-\rho D_t^z}\mid z]
\]
The Laplace transform of the IG process is
\[
 \E[e^{-qY^{(\k)}_u}]=e^{-u\psi_\k(q)},\qquad
 \psi_\k(q)=\k\left(\sqrt{1+2q/\k}-1\right),\qquad q\ge0
\]
and hence the survival function and first-passage CDF are
\begin{equation}
 S(t\mid z)=\exp\{-\rho C_z(t)-\Lambda_z(t)\psi_\k(\rho)\},\qquad
 F(t\mid z)=1-S(t\mid z)
 \label{eq:expig_survival}
\end{equation}
This is the IG specialization of the monotone-damage identity in Remark~\ref{rem:monotoneidentification}. If $a_z=0$, the first term in the exponent becomes $-\int_0^t\mu_0(u)\,\df u$, so accumulated damage modifies the survival function associated with the baseline force of mortality. No specific functional form for $\mu_0$ or $\Lambda_z$ is required. If recovery makes $C_z$ decrease, however, earlier boundary crossings must also be considered, and this monotonicity-based expression no longer applies in general.

\paragraph{Fixed initial energy.}
For $\xi=x_0>0$, replace averaging over the exponential distribution by
\begin{equation}
 S(t\mid x_0,z)=\P\{D_t^z<x_0-C_z(t)\mid z\}
 \label{eq:fixedig_threshold}
\end{equation}
For $r_z(t)=x_0-C_z(t)>0$ and $L_z(t)=\Lambda_z(t)>0$, the IG distribution gives
\begin{align}
 S(t\mid x_0,z)
 &=\Phi\!\left(\sqrt{\frac\k{r_z(t)}}[r_z(t)-L_z(t)]\right)\notag\\
 &\quad+e^{2\k L_z(t)}\Phi\!\left(-\sqrt{\frac\k{r_z(t)}}[r_z(t)+L_z(t)]\right),
 \label{eq:fixedig_survival}\\
 F(t\mid x_0,z)&=1-S(t\mid x_0,z)
 \label{eq:fixedig_cdf}
\end{align}
where $\Phi$ is the standard normal CDF. Survival is zero if $r_z(t)\le0$ and equals one if $r_z(t)>0$ and $L_z(t)=0$. In the empirical IG models, we set $C_z=0$ and represent all damage accumulation through $\Lambda_z$, so that $r_z(t)=x_0$.

\subsection{ID models: time-dependent drift and diffusion}\label{sec:fixedid_distribution}
The diffusion specification corresponding to the SEM of \citet{SMI2021} is
\begin{equation}
 \df X_t^z=b_z(t)\,\df t+\s_z(t)\,\df W_t,\qquad
 b_z(t)=-\mu_0(t)/\rho+a_z(t),\qquad
 X_0^z=\xi\sim\operatorname{Exp}(\rho)
 \label{eq:generalid_model}
\end{equation}
The drift $b_z$ combines baseline depletion $\mu_0/\rho$ and the adjustment $a_z$, while $\s_z$ controls random fluctuations. No sign restriction is imposed on $b_z$, allowing both depletion and recovery. Conditional on the environmental path, define cumulative drift and variance by
\[
 M_z(t)=\int_0^t b_z(u)\,\df u,\qquad
 A_z(t)=\int_0^t\s_z(u)^2\,\df u
\]
We impose Assumption~\ref{ass:reg} in Appendix~\ref{app:diffusion}. Piecewise continuous coefficients, including annual inputs, are covered under the conditions stated in Appendix~\ref{sec:diffusiondistribution}.

Unlike monotone IG damage, the unabsorbed diffusion may return above zero after a crossing. Its value at age $t$ therefore does not determine whether death has already occurred. We use Corollary~\ref{cor:mixedfortet}, obtained by averaging Theorem~\ref{thm:fortet} over the exponential initial distribution. The CDF of $\tau_z=\inf\{t:X_t^z\le0\}$ satisfies
\begin{equation}
 Q_{\rho,z}(t)=\int_{[0,t]}K_z(t,u)\,\df F(u\mid z),\qquad
 K_z(t,u)=\Phi\!\left(-\frac{M_z(t)-M_z(u)}{\sqrt{A_z(t)-A_z(u)}}\right),\quad u<t.
 \label{eq:expid_fortet}
\end{equation}
with $K_z(t,t)=1/2$. The left-hand side is the endpoint probability averaged over initial energy,
\begin{align*}
 Q_{\rho,z}(t)
 &=\Phi\!\left(-\frac{M_z(t)}{\sqrt{A_z(t)}}\right)\\
 &\quad-e^{\rho M_z(t)+\rho^2A_z(t)/2}
 \Phi\!\left(-\frac{M_z(t)+\rho A_z(t)}{\sqrt{A_z(t)}}\right)
\end{align*}
For general coefficients, the integral equation is solved numerically. Appendix~\ref{app:integral} gives the derivation, and Appendix~\ref{app:quadrature} describes the discretization.

\paragraph{Fixed initial energy.}
For a fixed initial value $x_0>0$, Theorem~\ref{thm:fortet} applies directly. The left-hand side is replaced by the corresponding fixed-initial-value endpoint probability, giving
\begin{equation}
 \Phi\!\left(-\frac{x_0+M_z(t)}{\sqrt{A_z(t)}}\right)
 =\int_{[0,t]}K_z(t,u)\,\df F(u\mid x_0,z)
 \label{eq:fixedid_fortet}
\end{equation}
for the same drift and diffusion coefficients. The kernel $K_z$ does not depend on the initial distribution. We use this equation in the empirical ID analysis. With fixed initial energy there is no exponential rate parameter; the net drift $b_z$ is specified directly.

\subsection{Conditioning on survival to entry and annual death probabilities}\label{sec:usedmodel_observation}
Mortality data need not follow individuals from birth. When observation begins at age 20, for example, death probabilities must be conditioned on survival to that age. This applies to both random and fixed initial energy. For brevity, we also write $F_c(t\mid z)$ and $S_c(t\mid z)$ in the fixed-initial-value case, suppressing the specified initial value. The probability of death by age $t$, conditional on survival to entry age $e_c$, is
\[
 F_c(t\mid\tau_{c,z}>e_c,z)
 =1-\frac{S_c(t\mid z)}{S_c(e_c\mid z)}
 =\frac{F_c(t\mid z)-F_c(e_c\mid z)}{1-F_c(e_c\mid z)},\qquad t\ge e_c
\]
provided $S_c(e_c\mid z)>0$. Conditioning retains the damage accumulated from birth and selects survivors; it does not reset energy at entry. The annual death probability for an individual alive at age $a$ is
\[
 q_{c,a}(z)=1-\frac{S_c(a+1\mid z)}{S_c(a\mid z)}
\]
For the fixed-initial-energy IG and ID models, these quantities are obtained from \eqref{eq:fixedig_survival} and \eqref{eq:fixedid_fortet}, respectively. Ratios can be unstable when survival probabilities are extremely small; annual probabilities whose numerical accuracy cannot be verified are treated as missing. Estimation uses squared errors in the conditional CDF, while annual death probability errors are reported as an additional measure.

\section{IG and ID specifications for Japanese mortality}\label{sec:paperig}
The empirical specifications were selected through extensive numerical comparisons. These explored initial energy distributions, damage and diffusion coefficients, the placement of covariate effects, information sharing across cohorts, estimation objectives and extrapolation rules. The resulting sex-specific fixed-initial-energy IG models combine competitive predictive performance with interpretable coefficients. We also consider a fixed-initial-energy ID model with an age change point, allowing for cohort effects and alternative estimation objectives.

This section presents the selected models and their estimation, rather than an exhaustive catalogue of all specifications explored. Section~\ref{sec:selectedresults} reports their predictive performance.
We fix the energy scale at $x_0=1000$ for IG models, following \citet{SSKMI2023}, and at $x_0=1$ for ID models. These are choices of units: multiplying initial energy, drift, diffusion coefficients and jump sizes by the same positive constant leaves the first-passage time unchanged.

\subsection{Mortality data and two environmental covariates}
Let $y$ denote calendar year and $x$ the age recorded in the data. We use national sex-specific deaths $D_{x,y}$ and person-years of exposure $E_{x,y}$ for Japan over 1947--2024 from the Japanese Mortality Database (JMD)~\citep{JMDnational}. For recorded ages $x=20,\ldots,89$, central death rates $m_{x,y}=D_{x,y}/E_{x,y}$ are converted to annual death probabilities as $q_{x,y}=1-e^{-m_{x,y}}$. Multiplying annual survival probabilities along cohort diagonals $c=y-x$ yields a mortality CDF for each sex and cohort. The terminal CDF is not rescaled to one, since some individuals remain alive at the end of observation. To exclude cohorts with limited follow-up, we remove three cohorts at each end of the training birth-year range. This rule (clip=3) is applied to every comparison model.

The covariates are real GDP per capita and physician density. Each series is log-transformed and standardized using only the mean and standard deviation of the training years, giving $\boldsymbol Z_y=(Z_{1,y},Z_{2,y})^\top$. National annual values are shared across sexes, but response coefficients are estimated separately.
\begin{table}[htbp]\centering\small
\caption{Input series and imputation for the selected models.}\label{tab:paperig_inputs}
\begin{tabular}{p{.18\linewidth}p{.71\linewidth}}\toprule Series & Coverage and preprocessing\\\midrule
Deaths and exposure & JMD, 1947--2024, by sex, recorded ages 20--89. Estimation uses data through the prediction origin.\\
Real GDP per capita & WDI, 1960--2024, spliced to a 1955--1959 series. Values before 1955 are set to the 1955 level. Log values are standardized over 1947 through the origin.\\
Physician density & Published 1955--2024 counts of physicians working in medical facilities per 100,000 population. Training values are linearly interpolated using observations through the origin; endpoints are extended by their nearest available value. Log values are standardized over the training period.\\\bottomrule
\end{tabular}\end{table}
Early GDP values are obtained by splicing historical series for this study; the resulting input is not a single official series constructed on a uniform basis~\citep{WDIgdp,CAOlong}. Physician density counts physicians working in medical facilities rather than all registered physicians~\citep{Phys2024series}. Training-period interpolation uses no physician observations beyond the prediction origin. The earliest value is carried backward and the last available value forward, making early-life environmental effects sensitive to imputation. During the prediction period, realized annual GDP and physician density are used. We therefore assess mortality prediction conditional on the environmental path, without forecasting the covariates themselves. The series also incorporate later revisions, so this exercise differs from prediction using only information available in real time.

\subsection{Male IG model: separate responses of baseline and age-related damage}
For cohort $c$, let
\[
X_c^z(a)=1000-D^{(c)}(\Lambda_{c,z}(a))
\]
where $D^{(c)}(s)$ is a nondecreasing IG process with mean $s$ and variance $s/\kappa_c$. The parameter $\kappa_c$ is constant over time within a cohort.
For males, we specify
\begin{align}
 \Lambda'_{c,z}(a)&=e^{g(c)}\left[b e^{-\boldsymbol\gamma_b^\top\boldsymbol Z_{c+a}}
       +\alpha e^{\alpha a}e^{-\boldsymbol\gamma_a^\top\boldsymbol Z_{c+a}}\right],
       \quad\Lambda_{c,z}(0)=0,\label{eq:selectedmale}\\
 g(c)&=ku_c+B(c)^\top\boldsymbol\theta,\quad
 u_c=c-\bar c_{\rm train},\quad
 \kappa_c=\kappa e^{vu_c}
\end{align}
We call the age-independent component \emph{baseline damage} and the increasing component \emph{age-related damage}; both depend on the environment. The parameters $\alpha,b,\kappa$ are positive. Here $b$ is the baseline damage coefficient, distinct from the net drift $b_z$ in the general model, and $\alpha e^{\alpha a}$ governs age-related accumulation. The vectors $\boldsymbol\gamma_b,\boldsymbol\gamma_a\in[0,2]^2$ describe responses to the same GDP and physician-density inputs. The smooth cohort effect $g(c)$ changes the rate of accumulation: a larger value implies greater expected damage, other things equal. The parameter $v$ controls cohort variation in damage variance for a fixed mean. A Gompertz-shaped time change does not imply a Gompertz--Makeham first-passage hazard for IG damage.

Here $\bar c_{\rm train}$ is the mean training birth year, so $u_c$ measures deviations in years. The coefficients $k$ and $v$ are therefore annual birth-year effects. The spline coefficient vector $\boldsymbol\theta$ is only one component of the full parameter collection denoted by $\th$.

The natural cubic spline $B(c)$ uses five equally spaced knots over the training birth-year range. After removing constant and linear components, the three remaining basis functions are scaled to unit root mean square. We impose the curvature penalty $10\|R\boldsymbol\theta\|^2$, where $R^\top R$ is the numerically integrated matrix of products of second derivatives. Bounds are $\boldsymbol\theta\in[-1,1]^3$, $k\in[-0.1,0.1]$ and $v\in[-0.05,0.05]$. The additional penalty $10^{-4}\|\boldsymbol\gamma_a-\boldsymbol\gamma_b\|^2$ discourages excessive differences between the two environmental responses. The model has 12 parameters, jointly estimated across all male training cohorts.

\subsection{Female IG model: age-dependent exposure to a period effect}
For females, the same IG first-passage distribution is used with
\begin{align}
 \Lambda'_{c,z}(a)&=(b+\alpha e^{\alpha a})
 \exp\{ku_c-\boldsymbol\gamma(a)^\top\boldsymbol Z_{c+a}
                 +\ell(a)r_{c+a}\},\label{eq:selectedfemale}\\
 \boldsymbol\gamma(a)&=(1-h(a))\boldsymbol\gamma_y+h(a)\boldsymbol\gamma_o,
 \quad h(a)=\tfrac12\{1+\tanh((a-65)/10)\},\\
 \ell(a)&=1+\beta h(a),\quad \kappa_c=\kappa,\quad\beta\in[-0.75,0.75]
\end{align}
The vectors $\boldsymbol\gamma_y,\boldsymbol\gamma_o\in[0,2]^2$ govern environmental responses at younger and older ages. The weight $h(a)$ changes smoothly from zero to one around age 65. A common period effect $r_y$ captures annual variation not explained by GDP and physician density. Its age loading is $\ell(a)$; when $\beta=0$, all ages have the same loading. Unlike the male model, the mean damage effect is linear in birth year and $\kappa$ is common across cohorts. During training, $\boldsymbol r=Q\boldsymbol\eta$, where the columns of $Q$ form an orthonormal basis orthogonal to the constant and the two environmental series, with $\eta_j\in[-3,3]$. This removes exact overlap with those components, but does not guarantee full identification. Normalizing the younger-age loading to one fixes the scale of the period effect.

The vector $\Delta^2\boldsymbol r$ contains second differences $r_y-2r_{y-1}+r_{y-2}$.
We use the penalty $0.01\|\Delta^2\boldsymbol r\|^2+0.001\|\boldsymbol\gamma_o-\boldsymbol\gamma_y\|^2+0.0001\beta^2$ to control annual roughness and excessive differences in age responses. Including period effects gives 64--79 parameters, depending on the prediction origin, compared with 12 for males. At origin $o$, the period effect is extrapolated as
\begin{equation}
 r_{o+j}=r_o+1.25\left(\frac15\sum_{i=0}^{4}[r_{o-i}-r_{o-i-1}]\right)j,
 \qquad j=1,\ldots,5
 \label{eq:selectedstateforecast}
\end{equation}
This extends the average annual change over the most recent five years at 1.25 times its rate for the next five years. The multiplier applies to changes in the period effect, not to death probabilities. Its suitability for longer horizons requires separate assessment.

\subsection{ID benchmark: an age change point and environmental diffusion effects}\label{sec:empiricalid}
We also fit a model with continuously fluctuating survival energy to the same data. Fixing $x_0=1$ in the diffusion model of Section~\ref{sec:fixedid_distribution} and suppressing cohort indices gives
\begin{equation}
 dX^z(a)=\mu A_T(a)\,da+
 \sigma\sqrt{A_T(a)}e^{-\boldsymbol\gamma^\top z(a)}\,dW(a),
 \qquad A_T(a)=1+\eta(a-T)_+^p
 \label{eq:fixedid_sde}
\end{equation}
where $(u)_+=\max(u,0)$, $\mu<0$, $\sigma,\eta,p>0$, and $T$ is the change-point age. Constant depletion and fluctuations operate before $T$, with an additional age effect thereafter. Covariates enter the diffusion coefficient and are estimated separately from the age profile of the drift.

Conditional on the environmental path, cumulative drift and variance are
\begin{align}
 M(a)&=\mu\left[a+\frac{\eta}{p+1}(a-T)_+^{p+1}\right],\\
 V_z(a)&=\sigma^2\int_0^a A_T(u)e^{-2\boldsymbol\gamma^\top z(u)}\,du
 \label{eq:fixedid_variance}
\end{align}
The variance integral is evaluated interval by interval with annual covariates held constant within each year.
The CDF follows from \eqref{eq:fixedid_fortet} with $M_z=M$ and $A_z=V_z$. Cohort coefficients are specified as
\[ \mu_c=-\exp\{m+k_\mu u_c+B(c)^\top\boldsymbol\theta\},\qquad \sigma_c=\exp\{s+k_\sigma u_c\} \]
Here $m$ and $s$ are the baseline log absolute drift and log diffusion coefficient, while $k_\mu$ and $k_\sigma$ are linear birth-year effects. The spline $B$ is constructed as in the male IG model; setting $\boldsymbol\theta=0$ gives a linear cohort specification. Covariates act on diffusion through $\exp(-\boldsymbol\gamma^\top\boldsymbol Z)$ rather than on the IG damage accumulation rate. The numerical sizes of the two models' environmental coefficients are therefore not directly comparable.
The selected male ID model uses $T=50$, a natural spline for the cohort drift effect, and equal weights on entry-conditioned CDF errors and five-year prediction errors within the training period.
The selected female ID model has the same structure and loss weights but uses $T=60$.
Four candidate specifications varied the change-point age (50 or 60), the cohort drift representation and the use of within-training five-year prediction errors. These were also compared with the previous linear-cohort ID model. The spline penalty is $10\|R\boldsymbol\theta\|^2$, as in the IG specification. Linear and spline versions have 8 and 11 parameters, respectively. One model is selected per sex and used for all cohorts, rather than choosing a separate specification from each cohort's evaluation results.
Parameter bounds are $|\mu_{\rm base}|\in[10^{-6},0.02]$ for $\mu_{\rm base}=-e^m$, $\sigma_{\rm base}\in[e^{-7},e]$ for $\sigma_{\rm base}=e^s$, $\eta\in[e^{-12},e^6]$, $p\in[0.2,8]$, $k_\mu,k_\sigma\in[-0.15,0.15]$, and $\gamma_j\in[0,2]$. The Fortet equation is solved using two steps per year during estimation and eight for reported predictions; accuracy is checked against sixteen steps per year. This makes ID estimation more costly than evaluation of the explicit IG survival function.
\subsection{Estimation objective and computation}\label{sec:estimation}
Parameters are estimated by matching model and observed CDFs. Hats denote estimates. Conditional on survival to entry age $e_c$, with $S_\th(e_c\mid z^c)>0$, the model CDF is
\[
 F_{\th,c}(t\mid\tau>e_c,z^c)
 =\frac{F_\th(t\mid z^c)-F_\th(e_c\mid z^c)}{S_\th(e_c\mid z^c)},\qquad t\ge e_c
\]
which does not reset energy at entry. This CDF-fitting approach follows \citet{SMI2021}. Errors across ages within a cohort are correlated, so standard errors based on independent CDF observations are not used.

The implementation scales birth-year deviations in decades. For presentation, we use years and divide the corresponding estimates and bounds for $k,v,k_\mu,k_\sigma$ by ten, leaving fitted products and predictions unchanged. Let $L_E$ be the mean squared CDF error conditional on survival to $e_c=\max(20,1947-c)$. Let $L_W$ be the mean squared error in five-year conditional CDFs at rolling origins within training. The female IG objective is $L_E$ plus penalties; the male objective is $\tfrac12L_E+\tfrac12L_W$ plus penalties. For prediction origin $o$, male $L_W$ uses the 20 origins $o-24,\ldots,o-5$, whose five-year endpoints are all at or before $o$. No post-origin mortality is used for training. The term $L_E$ uses all training observations from 1947 through $o$. Differences between male and female performance therefore reflect both specification and estimation choices.

We use box-constrained L-BFGS-B optimization from two training-based initial candidates and retain the solution with the smaller penalized objective. Common bounds are $\alpha\in[0.005,0.25]$, $b\in[0.01,100]$ and $\kappa\in[10^{-8},0.1]$. The IG CDF is evaluated from \eqref{eq:fixedig_survival} using cohort-specific $\kappa_c,\Lambda_{c,z}$, $C_z=0$ and $x_0=1000$. With covariates constant within each year, integration is analytic for males and uses eight-point Gauss--Legendre quadrature for females. Checks include finite-difference comparisons of derivatives, finer quadrature and recovery of known special cases. Optimizer termination alone does not establish global optimality.

\section{Predictive performance and interpretation}\label{sec:selectedresults}
\subsection{Evaluation design}
We use prediction origins 2004, 2009, 2014 and 2019, fitting each model to data available through the origin and predicting the next five years. At each origin $o$, the target cohorts are $c=o-79,\ldots,o-50$ for each sex. Origin, cohort and prediction year jointly define an evaluation point, giving $4\times30\times5=600$ points per sex.

Recorded age $x$ corresponds to the model's interval endpoint $a=x+1$. Calendar year is therefore $y=c+a-1$, and age at the prediction origin is $a_{0,c}=o-c+1$. For survivors at that origin, we evaluate
\[
 \widehat F_{c,o}(a\mid z)=1-\widehat S_c(a\mid z)/\widehat S_c(a_{0,c}\mid z),
 \qquad a=a_{0,c}+1,\ldots,a_{0,c}+5
\]
as the conditional probability of death by age $a$. The primary performance measure is
\begin{equation}
 R_F=\left\{\frac14\sum_o\frac1{150}\sum_{c,h}
 [\widehat F_{c,o}(a_{0,c}+h\mid z)-F^{\rm obs}_{c,o}(a_{0,c}+h)]^2\right\}^{1/2}
 \label{eq:selectedscore}
\end{equation}
where $h=1,\ldots,5$ denotes the prediction horizon and $F^{\rm obs}_{c,o}$ is the observed CDF under the same survival conditioning. The root mean squared error $R_F$ aggregates squared errors over all 600 points. We also compute the annual death probability RMSE, $R_q$, on the same points.

The same periods informed model and prediction-rule selection. These results therefore constitute a development-stage evaluation, not an untouched final test. Although one specification is used across all cohorts of each sex, an independent assessment requires new data.

\subsection{Reduced-form benchmarks: CBD models M6 and M7}\label{sec:cbdcomparison}
The Cairns--Blake--Dowd (CBD) model represents age and calendar-time variation through common period factors~\citep{Cairns2006}. Extensions allow cohort effects and more flexible age profiles. The labels M6 and M7 follow \citet{Cairns2009}, whose comparison designated Lee--Carter and related models as M1--M3, a spline mortality-surface model as M4, the basic CBD model as M5, and its extensions as M6--M8. We select M6, with a cohort effect, and M7, which also includes a quadratic age term, to compare GSEM with two CBD specifications differing in age curvature. An exhaustive comparison of the CBD family, including M8, is beyond the scope of this study.

Let $x$ be age, $y$ calendar year, $c=y-x$ birth year, and $q_{x,y}$ the annual death probability. The basic two-factor CBD model (M5) is
\begin{equation}
 \log\frac{q_{x,y}}{1-q_{x,y}}
 =K_y^{(1)}+(x-\bar x)K_y^{(2)}
 \label{eq:cbd-basic}
\end{equation}
where $\bar x$ is the mean age in the fitted range. The factors $K_y^{(1)}$ and $K_y^{(2)}$ describe changes in mortality level and age slope. CBD is reduced-form in the sense that it specifies the mortality surface directly. GSEM instead specifies an energy process and derives a lifetime distribution from its first-passage time. Its structural interpretation concerns the depletion and fluctuation mechanisms posited within the model~\citep{SMI2021,SSKMI2023}.

M6 adds a cohort effect. M7 adds a quadratic age term, allowing the curvature of the logit mortality profile to vary over time~\citep{Cairns2009}. Writing $\eta_{x,y}$ for the mortality predictor on the chosen link scale gives
\begin{align}
 \mathrm{M6}:\quad \eta_{x,y}
 &=K_y^{(1)}+(x-\bar x)K_y^{(2)}+\zeta_{y-x},
 \label{eq:cbd-m6}\\
 \mathrm{M7}:\quad \eta_{x,y}
 &=K_y^{(1)}+(x-\bar x)K_y^{(2)}\notag\\
 &\quad+\{(x-\bar x)^2-s_x^2\}K_y^{(3)}+\zeta_{y-x},
 \label{eq:cbd-m7}
\end{align}
where $s_x^2$ is the average of $(x-\bar x)^2$ over the age range, and $\zeta_c$ is the cohort effect. The third factor in M7 describes period variation in age curvature. Period factors are shared across cohorts, while $\zeta_c$ links observations within a cohort.

Estimation uses the identifying constraints of \citet{Cairns2009} and the Poisson formulation with log link for central death rates, $\log m_{x,y}=\eta_{x,y}$~\citep{Villegas2018}. Parameters are estimated by maximum likelihood from deaths and exposures through origin $o$, separately by sex. Future period factors are extrapolated using the mean of the 19 annual increments in the latest 20 estimated factor values. Effects for the already observed target cohorts are held at their fitted values.

Predicted annual death probabilities are $\widehat q_{x,y}=1-\exp\{-\exp(\widehat\eta_{x,y})\}$, assuming a constant force of mortality within each age--year cell. For cohort $c$, the probability of death within $h$ years, conditional on survival through origin $o$, is
\begin{equation}
 \widehat F^{\mathrm{CBD}}_{c,o}(h)
 =1-\prod_{r=1}^{h}\bigl(1-\widehat q_{o+r-c,o+r}\bigr),
 \qquad h=1,\ldots,5
 \label{eq:cbd-cdf}
\end{equation}
where $r$ is years since the origin and age in calendar year $o+r$ is $o+r-c$.

Training data, ages 20--89, exclusion of three cohorts at each boundary, prediction origins and evaluation points are aligned with GSEM. CBD receives no environmental covariates and is estimated by likelihood rather than CDF least squares. The comparison therefore concerns each model together with its stated estimation and prediction procedure.

\subsection{Prediction errors and numerical checks}\label{sec:predictioncomparison}
Table~\ref{tab:selectedaggregate} reports aggregate RMSEs and the number and proportion of evaluation points at which each model has the smallest absolute error. Table~\ref{tab:selectedfolds} gives RMSEs by origin. Pointwise counts complement RMSE by measuring how often a prediction is closest, rather than the magnitude of its error. Because observations within cohorts are correlated, these proportions are not win probabilities from independent trials.

\begin{table}[htbp]\centering\small
\caption{Historical prediction evaluation over 600 points per sex. RMSEs and counts (percentages) of smallest absolute errors among four models are reported; minima and largest counts are bold. Numbers of origins with lower CDF RMSE than M6/M7 are: IG, males 2/4 and 2/4, females 3/4 and 0/4; ID, males 0/4 and 1/4, females 0/4 and 0/4.}\label{tab:selectedaggregate}
\setlength{\tabcolsep}{4pt}
\begin{tabular}{llrrrr}\toprule
 & & \multicolumn{2}{c}{Mortality CDF} & \multicolumn{2}{c}{Annual death probability}\\
Sex & Model & RMSE $R_F$ & Best count (\%) & RMSE $R_q$ & Best count (\%)\\\midrule
Male & IG & \textbf{0.002900666} & \textbf{317/600 (52.8\%)} & \textbf{0.001268440} & \textbf{296/600 (49.3\%)}\\
Male & ID & 0.004172846 & 45/600 (7.5\%) & 0.001944777 & 52/600 (8.7\%)\\
Male & M6 & 0.003092249 & 84/600 (14.0\%) & 0.001346952 & 84/600 (14.0\%)\\
Male & M7 & 0.003222181 & 154/600 (25.7\%) & 0.001405803 & 168/600 (28.0\%)\\
\midrule
Female & IG & 0.002023327 & 88/600 (14.7\%) & 0.000758390 & 84/600 (14.0\%)\\
Female & ID & 0.004466515 & 52/600 (8.7\%) & 0.001690519 & 51/600 (8.5\%)\\
Female & M6 & 0.002313694 & 167/600 (27.8\%) & 0.000851078 & 193/600 (32.2\%)\\
Female & M7 & \textbf{0.001630848} & \textbf{293/600 (48.8\%)} & \textbf{0.000695235} & \textbf{272/600 (45.3\%)}\\
\bottomrule\end{tabular}
\par\smallskip\begin{minipage}{.97\linewidth}\footnotesize Counts use unrounded stored predictions at each point. There were no missing values or ties. In general, ties would be credited to all tied models. Percentages are rounded to one decimal place.\end{minipage}
\end{table}
\begin{table}[htbp]\centering
\caption{Conditional CDF RMSE by origin; row minima are bold. IG improves on ID at 3/4 male and 4/4 female origins. Among all four models, the smallest errors occur for male IG at 2/4 origins, M6 at 1/4 and M7 at 1/4; for females, M6 is best at 1/4 and M7 at 3/4. All other counts are 0/4.}\label{tab:selectedfolds}
\begin{tabular}{llrrrr}\toprule Sex & Origin & IG & ID & M6 & M7\\\midrule
Male & 2004 & \textbf{0.0015252} & 0.0041240 & 0.0039877 & 0.0045119\\
Male & 2009 & \textbf{0.0021680} & 0.0043721 & 0.0029252 & 0.0028522\\
Male & 2014 & 0.0038699 & 0.0036860 & \textbf{0.0025091} & 0.0025702\\
Male & 2019 & 0.0034136 & 0.0044656 & 0.0027376 & \textbf{0.0025361}\\
Female & 2004 & 0.0017942 & 0.0021400 & 0.0021126 & \textbf{0.0016849}\\
Female & 2009 & 0.0022170 & 0.0042645 & \textbf{0.0018442} & 0.0020833\\
Female & 2014 & 0.0021133 & 0.0042386 & 0.0030581 & \textbf{0.0008868}\\
Female & 2019 & 0.0019429 & 0.0062504 & 0.0020485 & \textbf{0.0016350}\\
\bottomrule\end{tabular}\end{table}

Aggregate CDF RMSE is lowest for male IG and female M7. Male IG has approximately 6.2\% lower RMSE than M6 and 10.0\% lower than M7, but exceeds both at the 2014 and 2019 origins. Female IG has approximately 12.5\% lower RMSE than M6 and 24.1\% higher RMSE than M7. ID has larger aggregate RMSE than IG for both sexes. Since parameter counts and estimation procedures also differ, this ranking cannot be attributed solely to damage versus diffusion dynamics.

\FloatBarrier
The figures show conditional mortality CDFs and absolute errors for cohort $c=o-65$ at each origin. Errors are expressed as $100\lvert\widehat F_{c,o}(\cdot\mid z)-F^{\rm obs}_{c,o}\rvert$ in percentage points. Each legend applies to all panels of its figure.

\begin{figure}[!htbp]\centering
\includegraphics[width=.90\linewidth]{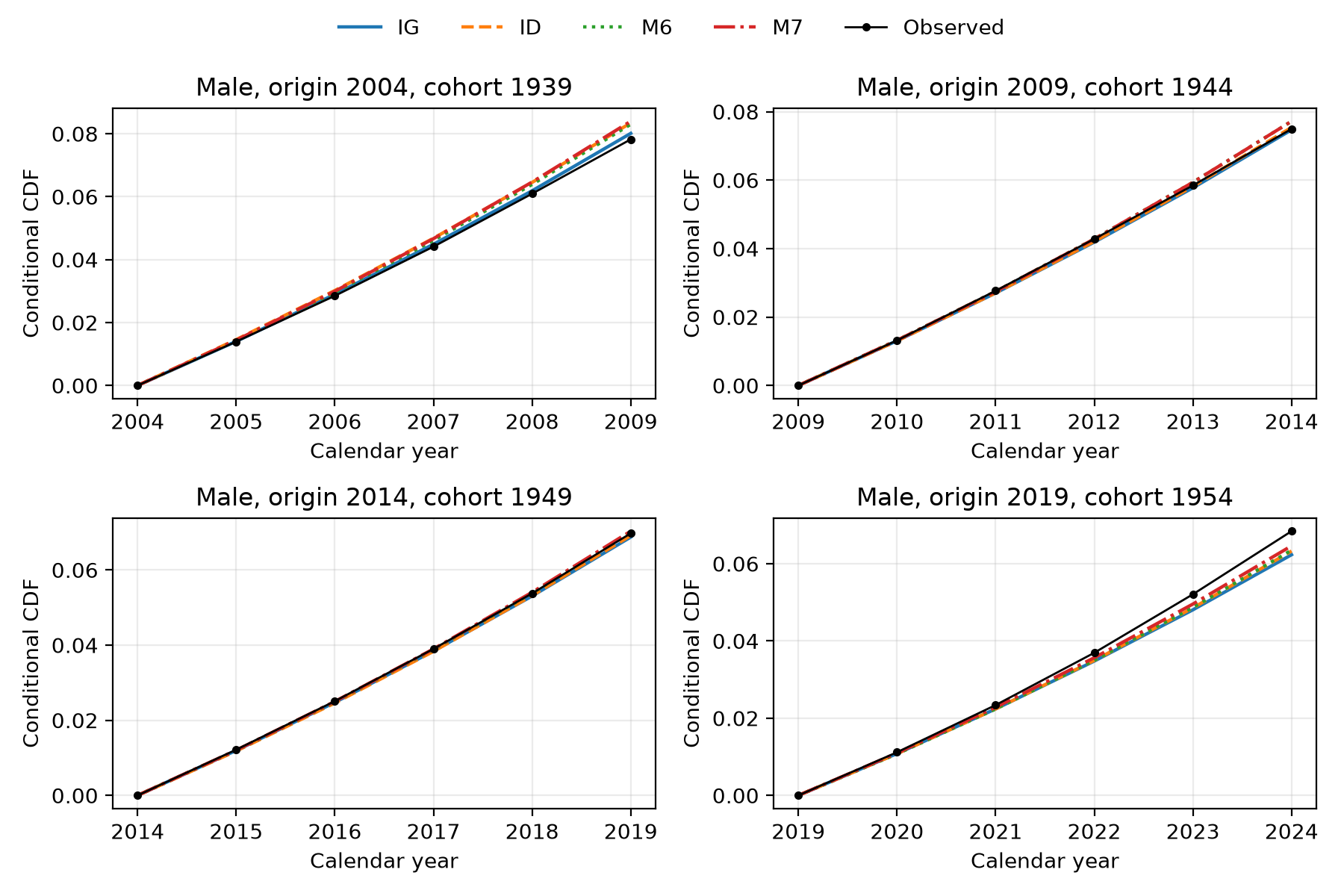}
\caption{Male mortality CDFs conditional on survival to the prediction origin, for cohort $c=o-65$ at each origin. Cohorts are selected by this rule, not by predictive performance. Conditional CDFs start at zero at the origin; subsequent values are predictions. The selected IG and ID specifications and the stated M6/M7 prediction procedures are used.}\end{figure}
\begin{figure}[!htbp]\centering
\includegraphics[width=.90\linewidth]{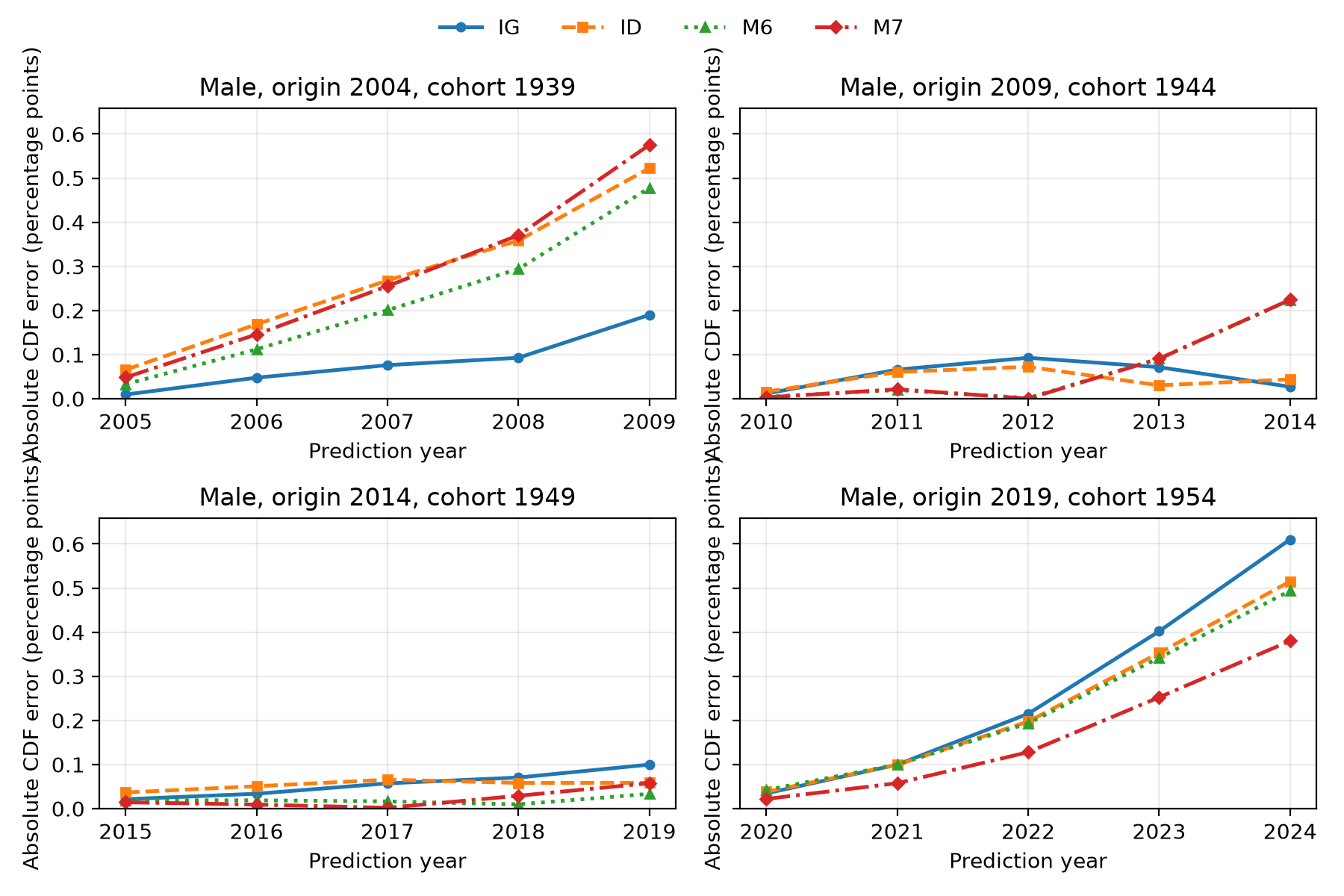}
\caption{Absolute errors in male conditional mortality CDFs (percentage points), using the same cohort rule, four origins and five prediction years as the preceding figure. Line styles and markers distinguish IG, ID, M6 and M7. Errors use the same predictions as the CDF plot.}\end{figure}
\begin{figure}[!htbp]\centering
\includegraphics[width=.90\linewidth]{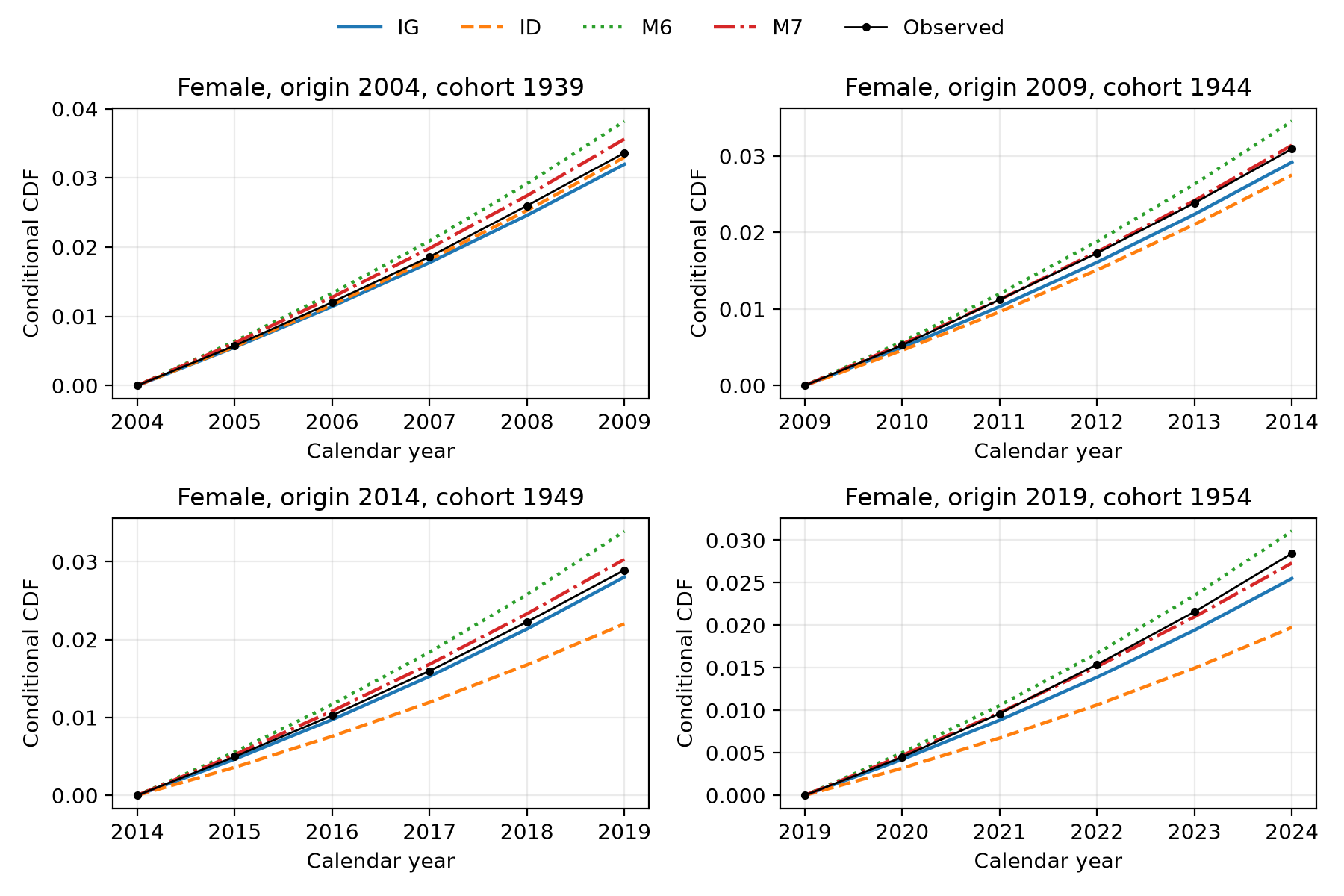}
\caption{Female mortality CDFs conditional on survival to the prediction origin, for cohort $c=o-65$ at each origin. Cohorts are selected by this rule, not by predictive performance. Conditional CDFs start at zero at the origin; subsequent values are predictions. The selected IG and ID specifications and the stated M6/M7 prediction procedures are used.}\end{figure}
\begin{figure}[!htbp]\centering
\includegraphics[width=.90\linewidth]{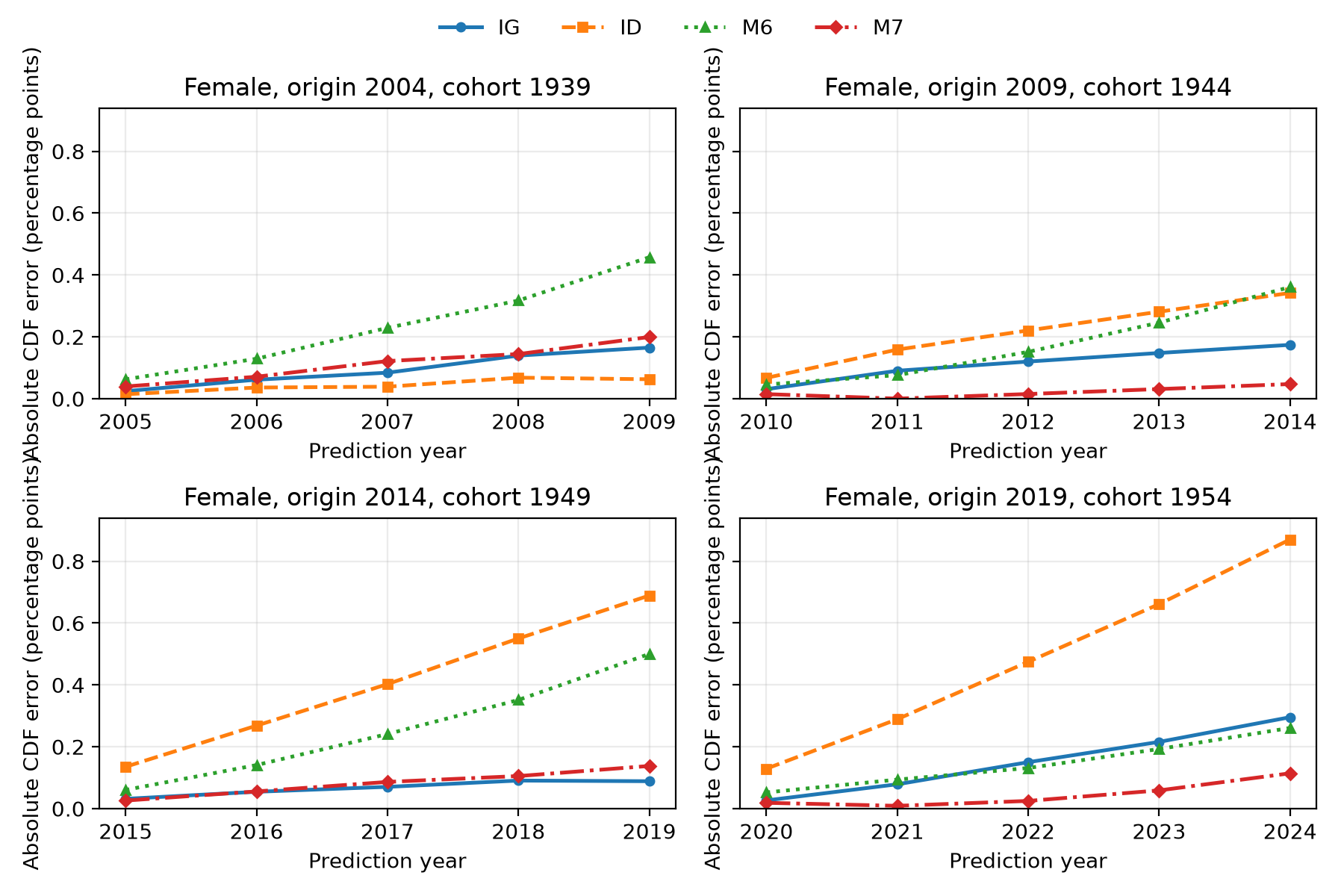}
\caption{Absolute errors in female conditional mortality CDFs (percentage points), using the same cohort rule, four origins and five prediction years as the preceding figure. Line styles and markers distinguish IG, ID, M6 and M7. Errors use the same predictions as the CDF plot.}\end{figure}
\FloatBarrier
Refining the ID grid from eight to sixteen steps per year changes the conditional CDF by at most $9.04\times10^{-5}$. IG predictions were also checked using finer integration and a separate implementation.

Across models, CDF RMSEs are approximately 0.16--0.45 percentage points and hence small on the probability scale. Many fitted curves are close visually. These results support the competitiveness of the IG specification against the CBD benchmarks, but do not establish a uniform advantage for GSEM. Differences remain, notably for female ID, and visual proximity alone does not imply equivalent predictive accuracy.

\subsection{Estimated damage and environmental responses}\label{sec:coefinterpretation}
\begin{table}[htbp]\centering
\caption{Selected IG parameter estimates at the 2019 origin, estimated jointly across cohorts of each sex. The coefficients $k,v$ use birth-year deviations measured in years.}
\begin{tabular}{lrr}\toprule Parameter & Male & Female\\\midrule
$\alpha$ & 0.068619 & 0.068915\\
$b$ & 9.784299 & 7.270629\\
$\kappa$ & 0.007144 & 0.004474\\
$k$ & -0.0020700 & -0.0072419\\
$v$ & -0.0018209 & 0.0000000\\
GDP: baseline/younger ages & 0.128054 & 0.141308\\
Physicians: baseline/younger ages & 0.000000 & 0.020404\\
GDP: age-related/older ages & 0.019179 & 0.052695\\
Physicians: age-related/older ages & 0.000000 & 0.004401\\
$\beta$ & -- & 0.091496\\
\bottomrule\end{tabular}\end{table}
A one-unit increase in standardized log GDP multiplies male IG baseline and age-related accumulation rates by approximately 0.880 and 0.981. The corresponding female multipliers are 0.868 at younger ages and 0.949 at older ages. Physician-density responses are smaller; male coefficients reach the lower bound of zero.

At the 2019 origin, ID coefficients for GDP and physician density are $(0.1310,0.1097)$ for males and $(0.1009,0.2324)$ for females. A one-unit standardized increase multiplies the diffusion coefficient by approximately 0.877 and 0.896 for males, and 0.904 and 0.793 for females, respectively. IG coefficients act on damage accumulation, whereas ID coefficients act on fluctuation amplitude. Their effects on mortality therefore cannot be compared from coefficient magnitudes alone. We next propagate changes in environmental paths through to mortality probabilities.

\subsection{Environmental scenarios and interpretation}\label{sec:envscenario}
Holding the coefficients estimated through 2019 fixed, we increase either GDP or physician density by 10\% in each year from 2020 to 2024. The baseline is the environmental path used in prediction evaluation. Only one covariate is changed at a time; past exposure, standardization, cohort effects and the female IG period effect remain fixed. The outcome is the five-year probability of death conditional on survival through the end of 2019.

Consider first the 1954 cohort, also shown in the CDF figures in Section~\ref{sec:predictioncomparison}. In the IG model, higher GDP reduces mean damage increments over the period by approximately 0.99\% for males and 1.05\% for females, with variance falling by the same proportions. Five-year death probabilities decline from 6.232\% to 6.165\% for males and from 2.550\% to 2.520\% for females. Mortality thus falls through slower damage accumulation. Higher physician density reduces the female probability from 2.550\% to 2.542\%, while the male probability is unchanged because its coefficients are zero.

In the ID model, higher physician density leaves the drift unchanged but reduces the diffusion coefficient by approximately 3.17\% for males and 6.60\% for females. Five-year death probabilities fall from 6.328\% to 6.121\% and from 1.975\% to 1.833\%, respectively. Here the reduction operates through smaller energy fluctuations rather than slower mean depletion. Whether lower diffusion reduces mortality depends on the drift and the age range considered.

\begin{figure}[htbp]
\centering
\includegraphics[width=\textwidth]{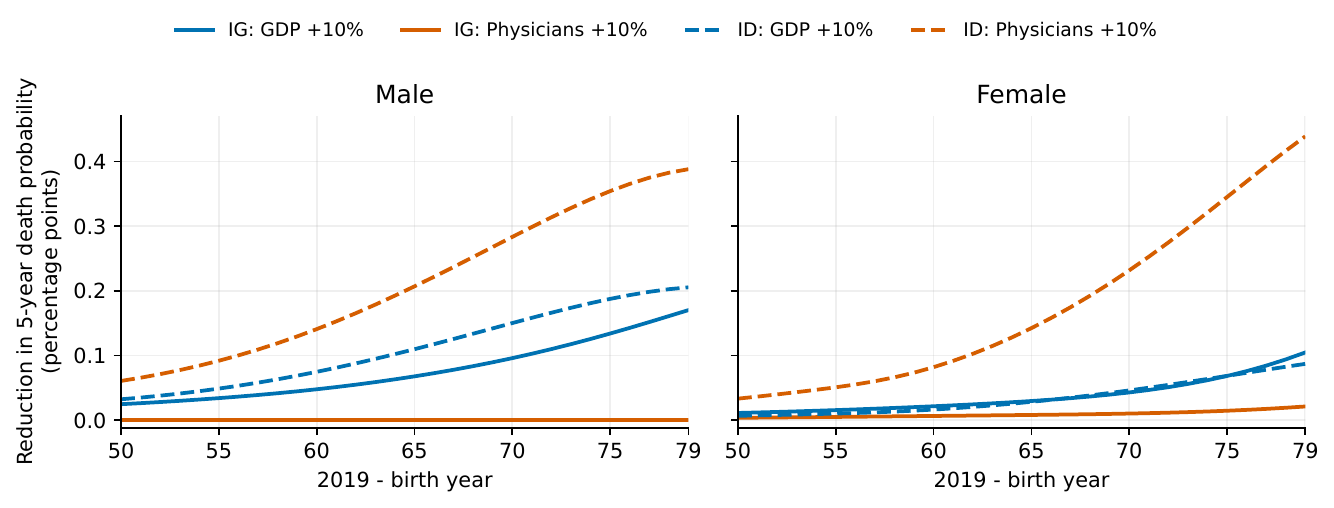}
\caption{Reductions in five-year death probabilities under alternative environmental paths. Baseline minus scenario probabilities are shown for 2020--2024, conditional on survival through 2019. GDP or physician density is raised by 10\% annually, holding fitted coefficients fixed. Males are on the left and females on the right; the legend is shared. Solid lines denote IG and dashed lines ID. The horizontal axis is 2019 minus birth year, with earlier cohorts to the right. Reductions are in percentage points.}
\label{fig:envscenario}
\end{figure}

Figure~\ref{fig:envscenario} extends the comparison to 30 cohorts per sex. In the IG model, reductions tend to be greater in older cohorts. Even a small response coefficient can produce a larger mortality change when baseline mortality or accumulated damage is greater. Age and cohort vary together in this figure, however, so these differences cannot be interpreted entirely as ageing effects.

The placement of ID covariates in diffusion was selected through predictive comparisons during development. Errors decreased relative to the covariate-free specification, while candidates that also included covariates in the drift did not improve performance. Within the specifications considered, this provides supporting evidence for environmental dependence in diffusion. Rankings vary by sex and evaluation period, however, and do not establish that environmental influences operate exclusively through diffusion.

These calculations trace a model-implied path from environmental changes through accumulated damage or fluctuations to mortality. This interpretability distinguishes GSEM from the direct mortality specification of CBD models. Environmental coefficients are constrained to be nonnegative, however, and parameter uncertainty is not quantified here. Responses also depend on the estimation period. A zero coefficient therefore does not imply that medical care is ineffective, nor are scenario reductions causal estimates of the benefits of environmental improvement.

\FloatBarrier

\section{Conclusion}
We have developed an environmental extension of survival energy models that derives cohort lifetime distributions from the exhaustion of latent energy. The framework accommodates exponential initial heterogeneity, diffusion and jumps. Numerical exploration with Japanese mortality data led to two fixed-initial-energy IG specifications: separate baseline and age-related environmental responses for males, and age-dependent environmental and period responses for females. A fixed-initial-energy ID model with an age change point was fitted to the same data for comparison.

The selected IG specifications have lower aggregate RMSE than the selected ID models for both sexes, supporting their use as the primary specifications for these data. The ranking need not extend to other populations or functional forms. Across the four origins, male IG has lower aggregate conditional CDF error than M6 and M7. Female IG improves with age-dependent responses but remains less accurate than M7. Even for males, the benchmarks perform better at some origins. Errors are small on the probability scale and many CDF curves are close, so model usefulness should be assessed through both predictive performance and the interpretation of predictions.

The structural contribution of GSEM is to link environmental variation to mean depletion or fluctuation amplitude. Scenario calculations can then identify model-implied conditions associated with lower mortality. The results do not establish that adding covariates always improves prediction, and fitted coefficients describe statistical associations rather than causal mechanisms. Performance is conditional on realized environmental paths and on the Poisson log-link implementation used for the CBD benchmarks.

Further work should evaluate the selected models on new data, including settings in which future environmental paths must themselves be predicted. Parameter uncertainty, extrapolation of the female period effect and early-life covariate imputation also require investigation. These issues are particularly relevant to longer-horizon forecasts of old-age mortality and the assessment of longevity risk.

\appendix
\renewcommand{\theequation}{\Alph{section}.\arabic{equation}}
\renewcommand{\thetheorem}{\Alph{section}.\arabic{theorem}}

\section{Diffusion notation and regularity}\label{app:diffusion}
This appendix collects the notation and assumptions used in Appendix~\ref{app:integral}.
Both appendices concern diffusions without jumps. Fix initial energy $\xi=h>0$ and an environmental path $Z=z$, and write $v_z(t):=\s_z(t)^2$. Then
\begin{equation}
 X_t^z=h+\int_0^t b_z(s)\,\df s+\int_0^t\s_z(s)\,\df W_s,
 \quad M_z(t)=\int_0^t b_z(s)\,\df s,\quad A_z(t)=\int_0^t v_z(s)\,\df s
 \label{eq:conditionaldiff}
\end{equation}
Unless stated otherwise, $F(t\mid h,z)$, $f(t\mid h,z)$, $S(t\mid h,z)$ and $p_{h,z}(t,x)$ refer to fixed initial energy $h$, with probability and expectation denoted by $\P^{h,z}$ and $\E^{h,z}$. Exponential initial energy is handled by integration in \eqref{eq:initialmix}. The coefficient $b_z$ is the net drift; in the baseline model without adjustment, $b_z=-\mu_0/\rho$.

\begin{assumption}[Regularity of conditional coefficients]\label{ass:reg}
For fixed $z$ and $\th$, the functions $b_z$ and $v_z$ are continuous on every finite interval $[0,T]$, and satisfy
\[
 0<\ul v_T\le v_z(t)\le\ol v_T<\infty,
 \qquad \sup_{0\le t\le T}|b_z(t)|<\infty
\]
Whenever classical PDE solutions or derivatives of densities are used, impose sufficient additional smoothness, such as $b_z,v_z\in C^1([0,T])$, to justify the relevant differentiation and interchange of integrals.
\end{assumption}
For piecewise coefficients, as with annual data, calculations are performed on each interval of continuity and joined at the break points. If coefficients are smoothed by interpolation, the resulting approximation error must also be assessed.

Conditional on the environmental path, the coefficients are deterministic. The It\^o integral $\int_0^t\s_z(s)\,\df W_s$ is therefore normal with mean zero and variance $A_z(t)$, giving
\begin{equation}
 X_t^z\sim N(h+M_z(t),A_z(t))\quad\text{under }\P^{h,z}.
 \label{eq:normal}
\end{equation}

\section{Integral equations for diffusion first-passage distributions}\label{app:integral}
Decomposing paths according to their first hitting time of zero yields an integral equation for the lifetime distribution. We derive that equation and state the conditions used in its numerical evaluation.
\subsection{A Fortet-type Volterra equation of the first kind}
The equation below is the Brownian boundary-crossing equation in Section~2 of \citet{JKV2009}, expressed in the original age variable of \eqref{eq:conditionaldiff}. Related numerical methods are discussed in Appendix~A of \citet{Smith2023}. We give a direct proof using the strong Markov property.

\begin{theorem}[Integral equation for the distribution function]\label{thm:fortet}
Under Assumption~\ref{ass:reg}, for $t>0$,
\begin{equation}
 Q_{h,z}(t)=\int_{[0,t]}K_z(t,s)\,\df F(s\mid h,z),\label{eq:fortet}
\end{equation}
where
\begin{align}
 Q_{h,z}(t)&=\Phi\!\l(-\frac{h+M_z(t)}{\sqrt{A_z(t)}}\r),\label{eq:H}\\
 K_z(t,s)&=\Phi\!\l(-\frac{M_z(t)-M_z(s)}{\sqrt{A_z(t)-A_z(s)}}\r)
 \quad(0\le s<t),\label{eq:K}\\
 K_z(t,t)&=\frac12.\nonumber
\end{align}
If $F(\cdot\mid h,z)$ has a density $f(\cdot\mid h,z)$, the right-hand side equals $\int_0^tK_z(t,s)f(s\mid h,z)\,\df s$.
\end{theorem}
\begin{proof}
Consider the unabsorbed process $X^z$. Continuity and $h>0$ imply that $X_t^z\le0$ entails $\tau_z\le t$. At a first hitting time $s<t$, $X_s^z=0$. By the strong Markov property of the time-augmented diffusion with deterministic coefficients, its transition distribution from that point to time $t$ is
\[
 N\bigl(M_z(t)-M_z(s),\ A_z(t)-A_z(s)\bigr)
\]
Decomposition at the first hitting time consequently gives
\[
 \P^{h,z}(X_t^z\le0)
 =\int_{[0,t]}\Phi\!\l(-\frac{M_z(t)-M_z(s)}{\sqrt{A_z(t)-A_z(s)}}\r)\,\df F(s\mid h,z).
\]
By \eqref{eq:normal}, the left-hand side is $Q_{h,z}(t)$. For fixed $t>0$, $\{\tau_z=t\}\subset\{X_t^z=0\}$, and the nondegenerate normal distribution of $X_t^z$ implies $\P^{h,z}(\tau_z=t)=0$. Also, $\tau_z>0$ almost surely. Endpoint values therefore do not affect the integral. Continuity of the coefficients gives $M_z(t)-M_z(s)=O(t-s)$ and $A_z(t)-A_z(s)\asymp t-s$, so $K_z(t,s)\to1/2$ as $s\uparrow t$.
\end{proof}

\begin{remark}[Endpoint versus first-passage probabilities]
The quantity $Q_{h,z}(t)$ is the probability that the unabsorbed process is nonpositive at time $t$, not the first-passage probability $F(t\mid h,z)$. Paths may return above zero after hitting the boundary. For example, if $b_z\equiv0$, then $K_z\equiv1/2$ and $F(t\mid h,z)=2Q_{h,z}(t)$.
\end{remark}

\subsection{General initial distributions}\label{sec:diffusiondistribution}
Averaging Theorem~\ref{thm:fortet} over initial energy accommodates individual heterogeneity without restricting the functional forms of $b_z$ and $v_z$.

\begin{corollary}[Mixtures over initial energy]\label{cor:mixedfortet}
Under Assumption~\ref{ass:reg}, let $\xi>0$ have distribution $\nu$ on $(0,\infty)$, independently of the Brownian motion and environment. Define
\[
 F(t\mid z)=\int_{(0,\infty)}F(t\mid h,z)\,\nu(\df h),\qquad
 Q_{\nu,z}(t)=\int_{(0,\infty)}Q_{h,z}(t)\,\nu(\df h)
\]
Then, with the same kernel $K_z$ as in Theorem~\ref{thm:fortet},
\begin{equation}
 Q_{\nu,z}(t)=\int_{[0,t]}K_z(t,u)\,\df F(u\mid z)
 \label{eq:mixedfortet}
\end{equation}
holds.
\end{corollary}
Integrate both sides of Theorem~\ref{thm:fortet} with respect to $\nu(\df h)$ and interchange the integrals by nonnegativity. The first-passage decomposition also applies to piecewise continuous coefficients if, on every finite interval, the drift is bounded and the variance coefficient is bounded above and away from zero.

\paragraph{Exponential initial energy.}
For $\nu(\df h)=\rho e^{-\rho h}\,\df h$, $\rho>0$, write $Q_{\rho,z}$ for $Q_{\nu,z}$. Integration by parts and completion of the square give
\begin{align}
 Q_{\rho,z}(t)
 &=\int_0^\infty\Phi\!\left(-\frac{h+M_z(t)}{\sqrt{A_z(t)}}\right)\rho e^{-\rho h}\,\df h\notag\\
 &=\Phi\!\left(-\frac{M_z(t)}{\sqrt{A_z(t)}}\right)
 -e^{\rho M_z(t)+\rho^2A_z(t)/2}
 \Phi\!\left(-\frac{M_z(t)+\rho A_z(t)}{\sqrt{A_z(t)}}\right)
 \label{eq:gmid_mixedbelow}
\end{align}
This is still an endpoint probability for the unabsorbed process, not a first-passage CDF. Substituting it into \eqref{eq:mixedfortet} and solving the integral equation yields the CDF. Appendix~\ref{app:quadrature} describes the numerical procedure.

The diffusion density equation is obtained by removing the jump term from \eqref{eq:PIDE} in Appendix~\ref{app:generaldistribution}. The initial distribution is a point mass at $h$ for fixed initial energy and has density $\rho e^{-\rho x}$ in the exponential case.

\subsection{Numerical solution}\label{app:quadrature}
Set $Q_z=Q_{h,z}$ for fixed initial energy and $Q_z=Q_{\nu,z}$ for a mixture. On a grid $0=t_0<\cdots<t_N$, let $s_j=(t_{j-1}+t_j)/2$ and $\Delta F_{j,z}=F(t_j\mid z)-F(t_{j-1}\mid z)$. Midpoint discretization of the first-kind equation gives
\[
 \Delta F_{n,z}\simeq
 \frac{Q_z(t_n)-\sum_{j<n}K_z(t_n,s_j)\Delta F_{j,z}}
 {K_z(t_n,s_n)},\qquad
 F(t_n\mid z)=\sum_{j=1}^n\Delta F_{j,z}
\]
For fixed initial energy, replace $F(t\mid z)$ by $F(t\mid h,z)$. Grid refinement is used to assess numerical error, together with checks that the CDF is nondecreasing and lies in $[0,1]$. The ID grids are specified in Section~\ref{sec:empiricalid}; numerical comparisons are reported in Section~\ref{sec:predictioncomparison}.

\section{First-passage representations for the general model}\label{app:generaldistribution}
We now consider the general jump-diffusion model containing the specifications of Section~3. Under monotone depletion, survival is determined by comparing cumulative loss with initial energy. With recovery or diffusion, earlier boundary crossings must also be accounted for. We extend the exponential-tail argument of Section~\ref{sec:randominitial} to the maximum cumulative loss.
Conditional on $z$, let $D^z$ be a nondecreasing process with independent increments, independent of the Brownian motion. Its jumps are described by a L\'{e}vy measure kernel $\Pi_{t,z}(dy)$ satisfying, for every finite $T$,
\[
 \int_0^T\int_{(0,\infty)}(1\wedge y)\Pi_{t,z}(dy)\,dt<\infty
\]
Deterministic depletion is included in the drift. In the finite-activity case, $\Pi_{t,z}(dy)=\lambda_z(t)G_{t,z}(dy)$, where $\lambda_z(t)$ is the arrival intensity and $G_{t,z}$ the distribution of jump sizes. For the IG model of Section~3,
\[
 \Pi_{t,z}(dy)=\Lambda'_z(t)\sqrt{\frac{\kappa}{2\pi}}\,y^{-3/2}e^{-\kappa y/2}\,dy,\qquad y>0
\]
The process has infinitely many small jumps but finite accumulated damage on finite intervals. Write $b_z(t)=-\mu_0(t)/\rho+a_z(t)$ and $v_z(t)=\sigma_z(t)^2$.

\subsection{A path representation with initial energy integrated out}
For a fixed path $z$, let $N_z(t)=\int_0^t\s_z(s)\,\df W_s$ and $M_z(t)=\int_0^t b_z(s)\,\df s$. Define the maximum cumulative loss by
\begin{equation}
 R_z(t)=\sup_{0\le s\le t}\{D_s^z-M_z(s)-N_z(s)\}\ge0
\end{equation}
Cumulative damage is finite and continuous components are bounded on finite intervals, so $R_z(t)$ is finite. The representation accommodates both finite and infinite jump activity, provided jumps represent downward energy losses.
\begin{proposition}[Survival function of the general model]\label{prop:survivalrepresentation}
Fix $z$ and assume exponential initial energy $\xi$ is independent of $(W,D^z)$. Then
\begin{equation}
 S_\th(t\mid z)=\E_\th^z[e^{-\rho R_z(t)}].\label{eq:expsup}
\end{equation}
Moreover, in terms of fixed-initial-energy distributions,
\begin{equation}
 F_\th(t\mid z)=\int_0^\infty F(t\mid h,z)\rho e^{-\rho h}\,\df h\label{eq:initialmix}
\end{equation}
holds.
\end{proposition}
\begin{proof}
Survival to time $t$ requires $\xi>D_s^z-M_z(s)-N_z(s)$ for all $s\le t$, equivalently $\xi>R_z(t)$ up to a null event under the continuous initial distribution.
Conditioning on $z,W,D^z$, applying the independent exponential tail and taking expectations proves \eqref{eq:expsup}.
Equation~\eqref{eq:initialmix} follows by conditioning on $\xi$.
\end{proof}
In general, $\P(X_t^z>0)$ is insufficient to determine survival, because the unabsorbed process may recover after an earlier boundary crossing.
Equation~\eqref{eq:expsup} is an exact expectation representation that analytically integrates out initial energy. It is useful for path simulation, but is not a closed-form solution because the expectation over paths remains to be evaluated. The mortality CDF is $F_\th(t\mid z)=1-S_\th(t\mid z)$.

\subsection{Monotone damage and identification}
\begin{remark}[Identifying latent damage mechanisms]\label{rem:monotoneidentification}
Without diffusion, if $C_z(t)=\int_0^t\{\mu_0(s)/\rho-a_z(s)\}\,\df s$ is nondecreasing, integration over independent exponential initial energy yields
\[
 S(t\mid z)=\exp\l\{-\rho C_z(t)-\int_0^t\Psi_{s,z}(\rho)\,\df s\r\},\qquad
 \Psi_{s,z}(q)=\int_{(0,\infty)}(1-e^{-qy})\Pi_{s,z}(\,\df y)
\]
This follows from the survival event $\{\xi>C_z(t)+D_t^z\}$ and the Laplace transform of damage. Under finite activity, $\Psi_{s,z}(\rho)=\la_z(s)\{1-\psi_{s,z}(\rho)\}$, where $\psi_{s,z}(\rho)=\E[e^{-\rho Y_s}\mid z]$ is the jump-size Laplace transform. Mortality distributions therefore do not separately identify arrival intensity and the jump-size distribution. For example, $a=0$, $\mu_0(t)=ge^{kt}$, constant intensity $\la$ and exponential jump sizes $Y\sim\operatorname{Exp}(\b)$ give the force of mortality $ge^{kt}+\la\rho/(\b+\rho)$. Distinct latent mechanisms can thus produce the same Gompertz--Makeham distribution. Good fit alone does not establish the biological reality of a mechanism. The simplification depends on monotonicity and generally fails when recovery is allowed.
\end{remark}

\subsection{The surviving-energy density and death-time density}
Conditional on $z$, the state variable $x$ is one-dimensional. Define the subprobability density of surviving energy by
\[
 \P_\th^z(X_t^z\in B,\tau_z>t)=\int_B p_z(t,x)\,\df x,
 \qquad B\subset(0,\infty)\text{ Borel}.
\]
Its integral is the survival probability $S_\th(t\mid z)$ rather than one. Absorption removes paths that have already died. If $v_z>0$, a sufficiently smooth density exists, and the required integrals and boundary limits are justified, then
\begin{align}
 \p_t p_z(t,x)
 &=-b_z(t)\p_xp_z(t,x)+\frac{v_z(t)}2\p_{xx}p_z(t,x)\notag\\
 &\quad+\int_{(0,\infty)}[p_z(t,x+y)-p_z(t,x)]\Pi_{t,z}(\,\df y),\quad x>0,
 \label{eq:PIDE}\\
 p_z(t,0)&=0\quad(t>0),\qquad p_z(0,x)=\rho e^{-\rho x}\quad(x>0)
 \label{eq:initialPDE}
\end{align}
This is a partial integro-differential equation (PIDE). The integral represents inflow from $x+y$ following a loss $y$, minus outflow from $x$. Under infinite activity these terms may diverge separately, so their difference must remain inside the integral. Under finite activity the term reduces to $\la_z[\int p_z(t,x+y)G_{t,z}(\,\df y)-p_z(t,x)]$.
The exponential initial density has a positive right limit at zero. The initial condition is therefore imposed in $L^1$ or in a weak sense, without requiring continuity at $(0,0)$.

\begin{proposition}[Forward equation and death-time density]\label{prop:forwarddeath}
Assume the stated regularity and integrability conditions and vanishing flux at infinity. For $\ol\Pi_{t,z}(x)=\Pi_{t,z}([x,\infty))$, $x>0$,
\begin{align}
 S_\th(t\mid z)&=\int_0^\infty p_z(t,x)\,\df x,\notag\\
 f_\th(t\mid z)&=\frac{v_z(t)}2\p_xp_z(t,0+)
 +\int_0^\infty p_z(t,x)\ol\Pi_{t,z}(x)\,\df x.
 \label{eq:jumpflux}
\end{align}
\end{proposition}
\begin{proof}
Extend a test function $\vp\in C_c^\infty((0,\infty))$ by zero on $x\le0$. The generator before absorption is
\[
 \A_{t,z}\vp(x)=b_z(t)\vp'(x)+\tfrac12v_z(t)\vp''(x)
 +\int[\vp(x-y)-\vp(x)]\Pi_{t,z}(\,\df y).
\]
Write $\I_A$ for the indicator of event $A$. Apply It\^o's formula with jumps up to the death time and take expectations of $\vp(X_t^z)\I_{\tau_z>t}$. The test function vanishes at absorption. Integration by parts and a change of variables yield the weak form of \eqref{eq:PIDE}.
Integrating over $x>0$, the continuous component gives boundary flux $-v_z\p_xp_z(t,0+)/2$.
For jumps restricted initially to $y>\ve$,
\[
 \begin{aligned}
 &\int_{(\ve,\infty)}\int_0^\infty
 [p_z(t,x+y)-p_z(t,x)]\,\df x\,\Pi_{t,z}(\,\df y)
 \\ &=-\int_{(\ve,\infty)}\int_0^y p_z(t,u)\,\df u\,\Pi_{t,z}(\,\df y)
 \end{aligned}
\]
Under the assumed integrability, letting $\ve\downarrow0$ and applying Tonelli's theorem gives
\[ -\int_0^\infty p_z(t,u)\ol\Pi_{t,z}(u)\,\df u. \]
The result follows from $-\partial_tS_\th(t\mid z)=f_\th(t\mid z)$.
\end{proof}
The first term represents deaths through continuous boundary hitting and the second deaths through jumps across the boundary.
When $\s=0$, the equation contains transport and jumps only. A zero boundary density cannot be imposed universally; the boundary condition must respect the direction of transport.
For example, outflow under negative drift contributes $-b_z(t)p_z(t,0+)$ to the death density when the boundary trace exists.
For monotone damage, the direct representation in Remark~\ref{rem:monotoneidentification} can be used.

\section*{Data and code availability}
The source data used in this study are available from the public sources cited in the article. The reproducibility package, including analysis code, input data, preparation instructions, and reference results, is publicly available on Zenodo (version 1.0.0): \url{https://doi.org/10.5281/zenodo.23519845}. The author's software, accompanying documentation, and original reference results are provided under the MIT License. Third-party data retain their original providers' terms, as documented in the package.

\section*{Declaration of competing interests}
The author declares no competing interests.

\section*{Declaration of generative AI and AI-assisted technologies in the manuscript preparation process}
During the preparation of this work, the author used OpenAI Codex to assist with data acquisition, code development, English translation and editorial revision of the manuscript. The author takes full responsibility for the content of the article.

\section*{Acknowledgements}
This work is partially supported by JSPS KAKENHI Grant-in-Aid for Scientific Research (C) (Grant No.~24K06875) and Japan Science and Technology Agency CREST (Grant No.~JPMJCR2115).


\begin{thebibliography}{99}
\bibitem[Cairns et al.(2006)]{Cairns2006}
Cairns, A. J. G., Blake, D. and Dowd, K. (2006).
A two-factor model for stochastic mortality with parameter uncertainty: Theory and calibration.
\emph{Journal of Risk and Insurance}, \textbf{73}(4), 687--718.
\href{https://doi.org/10.1111/j.1539-6975.2006.00195.x}{doi:10.1111/j.1539-6975.2006.00195.x}.
\bibitem[Cairns et al.(2009)]{Cairns2009}
Cairns, A. J. G., Blake, D., Dowd, K., Coughlan, G. D., Epstein, D., Ong, A. and Balevich, I. (2009).
A quantitative comparison of stochastic mortality models using data from England and Wales and the United States.
\emph{North American Actuarial Journal}, \textbf{13}(1), 1--35.
\bibitem[Dimai(2025)]{Dimai2025}
Dimai, M. (2025).
Modeling and forecasting mortality with economic, environmental and lifestyle variables.
\emph{Decisions in Economics and Finance}, \textbf{48}, 1161--1195.
\href{https://doi.org/10.1007/s10203-024-00434-4}{doi:10.1007/s10203-024-00434-4}.
\bibitem[Gompertz(1825)]{Gompertz1825}
Gompertz, B. (1825).
On the nature of the function expressive of the law of human mortality, and on a new mode of determining the value of life contingencies.
\emph{Philosophical Transactions of the Royal Society of London}, \textbf{115}, 513--583.
\href{https://doi.org/10.1098/rstl.1825.0026}{doi:10.1098/rstl.1825.0026}.
\bibitem[Jaimungal et al.(2009)]{JKV2009}
Jaimungal, S., Kreinin, A. and Valov, A. (2009).
Integral Equations and the First Passage Time of Brownian Motions.
Preprint, arXiv:0902.2569v2.
\url{https://arxiv.org/abs/0902.2569}.
\bibitem[Mitsuda and Shimizu(2024)]{MS2024}
Mitsuda, D. and Shimizu, Y. (2024).
Mortality prediction using survival energy models with functional data analysis.
\emph{Japanese Journal of Statistics and Data Science}, \textbf{7}, 841--859.
\href{https://doi.org/10.1007/s42081-024-00245-2}{doi:10.1007/s42081-024-00245-2}.
\bibitem[Molini et al.(2011)]{Molini2011}
Molini, A., Talkner, P., Katul, G. G. and Porporato, A. (2011).
First passage time statistics of Brownian motion with purely time dependent drift and diffusion.
\emph{Physica A: Statistical Mechanics and its Applications}, \textbf{390}(11), 1841--1852.
\href{https://doi.org/10.1016/j.physa.2011.01.024}{doi:10.1016/j.physa.2011.01.024}.
Preprint version: \url{https://arxiv.org/abs/1009.0557} (v3).
\bibitem[Niu and Melenberg(2014)]{NiuMelenberg2014}
Niu, G. and Melenberg, B. (2014).
Trends in mortality decrease and economic growth.
\emph{Demography}, \textbf{51}(5), 1755--1773.
\href{https://doi.org/10.1007/s13524-014-0328-3}{doi:10.1007/s13524-014-0328-3}.
\bibitem[OECD(n.d.)]{OECDhealth}
OECD. (n.d.). OECD Health Statistics.
\url{https://www.oecd.org/en/data/datasets/oecd-health-statistics.html}.
Accessed 29 September 2026.
\bibitem[Shimizu et al.(2021)]{SMI2021}
Shimizu, Y., Minami, Y. and Ito, R. (2021).
Why does a human die? A structural approach to cohort-wise mortality prediction under survival energy hypothesis.
\emph{ASTIN Bulletin}, \textbf{51}(1), 191--219.
\href{https://doi.org/10.1017/asb.2020.32}{doi:10.1017/asb.2020.32}.
\bibitem[Shimizu et al.(2023)]{SSKMI2023}
Shimizu, Y., Shirai, K., Kojima, Y., Mitsuda, D. and Inoue, M. (2023).
Survival energy models for mortality prediction and future prospects.
\emph{ASTIN Bulletin}, \textbf{53}(2), 377--391.
\href{https://doi.org/10.1017/asb.2023.10}{doi:10.1017/asb.2023.10}.
\bibitem[Smith(2023)]{Smith2023}
Smith, P. L. (2023).
``Reliable organisms from unreliable components'' revisited: the linear drift, linear infinitesimal variance model of decision making.
\emph{Psychonomic Bulletin \& Review}, \textbf{30}, 1323--1359.
\href{https://doi.org/10.3758/s13423-022-02237-3}{doi:10.3758/s13423-022-02237-3}.
\bibitem[Villegas et al.(2018)]{Villegas2018}
Villegas, A. M., Kaishev, V. K. and Millossovich, P. (2018).
StMoMo: An R package for stochastic mortality modeling.
\emph{Journal of Statistical Software}, \textbf{84}(3), 1--38.
\href{https://doi.org/10.18637/jss.v084.i03}{doi:10.18637/jss.v084.i03}.
\bibitem[World Bank(n.d.-a)]{WBdata}
World Bank. (n.d.-a). World Development Indicators, indicator pages and Indicators API.
Codes for the six series examined here are listed in Example~\ref{ex:dataavailability}.
\url{https://data.worldbank.org/indicator/SH.MED.PHYS.ZS};
\url{https://data.worldbank.org/indicator/NY.GDP.PCAP.PP.KD};
\url{https://data.worldbank.org/indicator/EN.ATM.PM25.MC.M3};
\url{https://data.worldbank.org/indicator/SH.XPD.CHEX.PP.CD};
\url{https://data.worldbank.org/indicator/SL.UEM.TOTL.ZS};
\url{https://data.worldbank.org/indicator/SH.STA.SMSS.ZS}.
Japanese API series retrieved 29 September 2026.
\bibitem[World Bank(n.d.-b)]{WBexpanded}
World Bank. (n.d.-b). World Development Indicators: Urban population (\texttt{SP.URB.TOTL.IN.ZS}), CO$_2$ emissions per capita (\texttt{EN.GHG.CO2.PC.CE.AR5}) and PM2.5 air pollution (\texttt{EN.ATM.PM25.MC.M3}).
\url{https://data.worldbank.org/indicator/SP.URB.TOTL.IN.ZS};
\url{https://data.worldbank.org/indicator/EN.GHG.CO2.PC.CE.AR5};
\url{https://data.worldbank.org/indicator/EN.ATM.PM25.MC.M3}.
Japanese series retrieved 6 October 2026.
\bibitem[World Bank(n.d.-c)]{WDIgdp} World Bank. (n.d.-c). GDP per capita (constant 2015 US\$), Japan, NY.GDP.PCAP.KD. \url{https://data.worldbank.org/indicator/NY.GDP.PCAP.KD?locations=JP}.
\bibitem[Zhu et al.(2024)]{ZZW2024}
Zhu, X., Zhou, K. Q. and Wang, Z. (2024).
A new paradigm of mortality modeling via individual vitality dynamics.
arXiv:2407.15388v2.
\url{https://arxiv.org/html/2407.15388v2}.
\bibitem[Zhu et al.(2026)]{ZZW2026}
Zhu, X., Zhou, K. Q. and Wang, Z. (2026).
Mortality modeling via vitality: Model constructions and actuarial applications.
\emph{Insurance: Mathematics and Economics}, \textbf{126}, 103181.
\href{https://doi.org/10.1016/j.insmatheco.2025.103181}{doi:10.1016/j.insmatheco.2025.103181}.
\bibitem[Japan Meteorological Agency(n.d.)]{JMAdata}
Japan Meteorological Agency. (n.d.). Historical weather data download: User guide (in Japanese).
\url{https://www.data.jma.go.jp/risk/obsdl/top/help1.html}.
Accessed 29 September 2026.
\bibitem[MHLW(n.d.-a)]{MedExpenditure}
Ministry of Health, Labour and Welfare (MHLW). (n.d.-a). National medical care expenditure, fiscal year 2023 (in Japanese).
\url{https://www.mhlw.go.jp/toukei/saikin/hw/k-iryohi/23/dl/R05data.pdf}. Accessed 6 October 2026.
\bibitem[MHLW(n.d.-b)]{ESTATdoctors}
Ministry of Health, Labour and Welfare; e-Stat (MHLW). (n.d.-b). Statistics of Physicians, Dentists and Pharmacists, 2024: Physicians (in Japanese).
\url{https://www.e-stat.go.jp/stat-search/files?cycle=7&tclass=000001234962&year=20240}.
Accessed 29 September 2026.
\bibitem[MHLW(n.d.-c)]{NHNS}
Ministry of Health, Labour and Welfare (MHLW). (n.d.-c). National Health and Nutrition Survey (in Japanese).
\url{https://www.mhlw.go.jp/bunya/kenkou/kenkou_eiyou_chousa.html}.
Accessed 29 September 2026.
\bibitem[MHLW(n.d.-d)]{MHLWlong}
Ministry of Health, Labour and Welfare (MHLW). (n.d.-d). Survey of Physicians, Dentists and Pharmacists, 2000, Table 16: Historical series (in Japanese).
\url{https://www.mhlw.go.jp/toukei/saikin/hw/ishi/00/tou16.html}.
\bibitem[MHLW(n.d.-e)]{Phys2024series} Ministry of Health, Labour and Welfare (MHLW). (n.d.-e). Statistics of Physicians, Dentists and Pharmacists, 2024, Historical Table 3 (in Japanese). \url{https://www.e-stat.go.jp/stat-search/files?stat_infid=000040383721}.
\bibitem[MHLW(n.d.-f)]{WageSeries} Ministry of Health, Labour and Welfare (MHLW). (n.d.-f). Monthly Labour Survey, national long-term series: Annual percentage changes in real total cash earnings, establishments with 30 or more employees, all surveyed industries and employment types (in Japanese). \url{https://www.e-stat.go.jp/stat-search/files?stat_infid=000032189777}. Retrieved 29 September 2026.
\bibitem[MHLW(n.d.-g)]{WageCodes} Ministry of Health, Labour and Welfare (MHLW). (n.d.-g). Guide to symbols in the accumulated actual-value and index data (in Japanese). \url{https://www.e-stat.go.jp/stat-search/files?stat_infid=000032189778}.
\bibitem[MHLW(n.d.-h)]{WageMethods} Ministry of Health, Labour and Welfare (MHLW). (n.d.-h). Monthly Labour Survey: Revisions to indices and rates of change (in Japanese). \url{https://www.mhlw.go.jp/toukei/list/30-1c.html}.
\bibitem[IPSS(n.d.-a)]{JMDdata}
National Institute of Population and Social Security Research (IPSS). (n.d.-a). Japanese Mortality Database (in Japanese).
\url{https://www.ipss.go.jp/p-toukei/JMD/index.asp}.
Accessed 29 September 2026.
\bibitem[IPSS(n.d.-b)]{JMDnational}
National Institute of Population and Social Security Research (IPSS). (n.d.-b). Japanese Mortality Database, national data, version 005\_002 (1947--2024) (in Japanese).
\url{https://www.ipss.go.jp/p-toukei/JMD/00/}. Accessed 29 September 2026.
\bibitem[Statistics Bureau of Japan(n.d.-a)]{LFSlong}
Statistics Bureau of Japan. (n.d.-a). Labour Force Survey: Long-term time series (in Japanese).
\url{https://www.stat.go.jp/data/roudou/longtime/03roudou.html}.
\bibitem[Statistics Bureau of Japan(n.d.-b)]{LFSqa}
Statistics Bureau of Japan. (n.d.-b). Labour Force Survey: Questions and answers (in Japanese).
\url{https://www.stat.go.jp/data/roudou/qa-1.htm}.

\bibitem[Statistics Bureau of Japan(n.d.-c)]{LFSdownload} Statistics Bureau of Japan. (n.d.-c). Labour Force Survey, Long-term Table 1 a-1: Unadjusted monthly observations (in Japanese). \url{https://www.e-stat.go.jp/stat-search/files?stat_infid=000031831358}.
\bibitem[Cabinet Office of Japan(n.d.)]{CAOlong}
Cabinet Office of Japan. (n.d.). Long-term economic statistics: National accounts (in Japanese).
\url{https://www5.cao.go.jp/keizai3/2017/0118nk/n17_7_data01.html}.
\bibitem[Minami(1987)]{Minami1987}
Minami, Y. (1987). Mortality and morbidity rates based on stochastic models (in Japanese; title translated).
\emph{Bulletin of the Institute of Actuaries of Japan}, No. 40, Part 1, 31--53.
\end{thebibliography}
\end{document}